\documentclass{lmcs}
\pdfoutput=1

\usepackage{lastpage}
\lmcsdoi{16}{4}{19}
\lmcsheading{}{\pageref{LastPage}}{}{}%
{Sep.~27,~2017}{Dec.~22,~2020}{}

\pdfoutput=1

\usepackage[utf8]{inputenc}

\usepackage{amssymb,nicefrac,amsthm}
\usepackage{amsmath}
\usepackage{stmaryrd}

\usepackage{tikzfig}

\tikzstyle{env}=[copoint,regular polygon rotate=0,minimum width=0.2cm, fill=black]

\tikzstyle{hadamard edge}=[-,color=blue,dashed,dash pattern=on 2pt off 0.7pt]
\tikzstyle{black edge}=[-,color=black,dashed,dash pattern=on 2pt off 0.7pt]

\tikzstyle{every picture}=[baseline=-0.25em]
\tikzstyle{dotpic}=[scale=0.5]
\tikzstyle{diredges}=[every to/.style={diredge}]
\tikzstyle{dot graph}=[shorten <=-0.1mm,shorten >=-0.1mm,scale=0.6]
\tikzstyle{plot point}=[circle,fill=black,minimum width=2mm,inner sep=0]

\tikzstyle{braceedge}=[decorate,decoration={brace,amplitude=2mm,raise=-1mm}]
\tikzstyle{small braceedge}=[decorate,decoration={brace,amplitude=1mm,raise=-1mm}]
\tikzstyle{left hook arrow}=[left hook-latex]
\tikzstyle{right hook arrow}=[right hook-latex]

\tikzstyle{black dot}=[inner sep=0.7mm,minimum width=0pt,minimum height=0pt,fill=black,draw=black,shape=circle]

\tikzstyle{dot}=[black dot]
\tikzstyle{smalldot}=[inner sep=0.4mm,minimum width=0pt,minimum height=0pt,fill=black,draw=black,shape=circle]
\tikzstyle{white dot}=[dot,fill=white]
\tikzstyle{antipode}=[white dot,inner sep=0.3mm,font=\footnotesize]
\tikzstyle{smallwhitedot}=[smalldot,fill=white]
\tikzstyle{alt white dot}=[white dot,label={[xshift=3.07mm,yshift=-0.05mm,font=\footnotesize]left:$*$}]
\tikzstyle{gray dot}=[dot,fill=gray!40!white]
\tikzstyle{smallgraydot}=[smalldot,fill=gray!40!white]
\tikzstyle{box vertex}=[draw=black,rectangle]
\tikzstyle{small box}=[box vertex,fill=white]
\tikzstyle{whitebg}=[fill=white,inner sep=2pt]
\tikzstyle{graph state vertex}=[sg vertex,fill=black]

\tikzstyle{wide copoint}=[fill=white,draw=black,shape=isosceles triangle,shape border rotate=90,isosceles triangle stretches=true,inner sep=1pt,minimum width=1.5cm,minimum height=5mm]
\tikzstyle{wide point}=[fill=white,draw=black,shape=isosceles triangle,shape border rotate=-90,isosceles triangle stretches=true,inner sep=1pt,minimum width=1.5cm,minimum height=4mm]
\tikzstyle{very wide copoint}=[fill=white,draw=black,shape=isosceles triangle,shape border rotate=-90,isosceles triangle stretches=true,inner sep=1pt,minimum width=2.5cm,minimum height=4mm]
\tikzstyle{very wide empty copoint}=[draw=black,shape=isosceles triangle,shape border rotate=-90,isosceles triangle stretches=true,inner sep=1pt,minimum width=2.5cm,minimum height=4mm]
\tikzstyle{symm}=[ultra thick,shorten <=-1mm,shorten >=-1mm]

\tikzstyle{square box}=[rectangle,fill=white,draw=black,minimum height=5mm,minimum width=5mm,font=\small]
\tikzstyle{square gray box}=[rectangle,fill=gray!30,draw=black,minimum height=6mm,minimum width=6mm]
\tikzstyle{copoint}=[regular polygon,regular polygon sides=3,draw=black,scale=0.75,inner sep=-0.5pt,minimum width=7mm,fill=white]
\tikzstyle{point}=[regular polygon,regular polygon sides=3,draw=black,scale=0.75,inner sep=-0.5pt,minimum width=7mm,fill=white,regular polygon rotate=180]
\tikzstyle{gray point}=[point,fill=gray!40!white]
\tikzstyle{gray copoint}=[copoint,fill=gray!40!white]

\newcommand{\edgearrow}{{\arrow[black]{>}}}
\newcommand{\edgetick}{{\arrow[black,scale=0.7,very thick]{|}}}

\tikzstyle{diredge}=[->]
\tikzstyle{rdiredge}=[<-]
\tikzstyle{medium diredge}=[->]

\tikzstyle{short diredge}=[->]
\tikzstyle{halfedge}=[-)]
\tikzstyle{other halfedge}=[(-]
\tikzstyle{freeedge}=[(-)]
\tikzstyle{white edge}=[line width=5pt,white]
\tikzstyle{tick}=[postaction=decorate,decoration={markings, mark=at position 0.5 with \edgetick}]
\tikzstyle{small map edge}=[|-latex, gray!60!blue, shorten <=0.9mm, shorten >=0.5mm]
\tikzstyle{thick dashed edge}=[very thick,dashed,gray!40]
\tikzstyle{map edge}=[|-latex,very thick, gray!40, shorten <=1mm, shorten >=0.5mm]
\tikzstyle{tickedge}=[postaction=decorate,
  decoration={markings, mark=at position 0.5 with \edgetick}]
\tikzstyle{dirtickedge}=[postaction=decorate,
  decoration={markings, mark=at position 0.5 with \edgetick},
  decoration={markings, mark=at position 0.85 with \edgearrow}]
\tikzstyle{dirdoubletickedge}=[postaction=decorate,
  decoration={markings, mark=at position 0.4 with \edgetick},
  decoration={markings, mark=at position 0.6 with \edgetick},
  decoration={markings, mark=at position 0.85 with \edgearrow}]

\makeatletter
\newcommand{\boxshape}[3]{%
\pgfdeclareshape{#1}{
\inheritsavedanchors[from=rectangle] 
\inheritanchorborder[from=rectangle]
\inheritanchor[from=rectangle]{center}
\inheritanchor[from=rectangle]{north}
\inheritanchor[from=rectangle]{south}
\inheritanchor[from=rectangle]{west}
\inheritanchor[from=rectangle]{east}
\backgroundpath{
\southwest \pgf@xa=\pgf@x \pgf@ya=\pgf@y
\northeast \pgf@xb=\pgf@x \pgf@yb=\pgf@y

\@tempdima=#2
\@tempdimb=#3

\pgfpathmoveto{\pgfpoint{\pgf@xa - 5pt + \@tempdima}{\pgf@ya}}
\pgfpathlineto{\pgfpoint{\pgf@xa - 5pt - \@tempdima}{\pgf@yb}}
\pgfpathlineto{\pgfpoint{\pgf@xb + 5pt + \@tempdimb}{\pgf@yb}}
\pgfpathlineto{\pgfpoint{\pgf@xb + 5pt - \@tempdimb}{\pgf@ya}}
\pgfpathlineto{\pgfpoint{\pgf@xa - 5pt + \@tempdima}{\pgf@ya}}
\pgfpathclose
}
}}

\boxshape{NEbox}{0pt}{8pt}
\boxshape{SEbox}{0pt}{-8pt}
\boxshape{NWbox}{8pt}{0pt}
\boxshape{SWbox}{-8pt}{0pt}
\makeatother

\tikzstyle{map}=[draw,shape=NEbox,inner sep=7pt]
\tikzstyle{mapdag}=[draw,shape=SEbox,inner sep=7pt]
\tikzstyle{maptrans}=[draw,shape=SWbox,inner sep=7pt]
\tikzstyle{mapconj}=[draw,shape=NWbox,inner sep=7pt]

\tikzstyle{probs}=[shape=semicircle,fill=gray!40!white,draw=black,shape border rotate=180,minimum width=1.2cm]

\tikzstyle{arrs}=[-latex,font=\small,auto]
\tikzstyle{arrow plain}=[arrs]
\tikzstyle{arrow dashed}=[dashed,arrs]
\tikzstyle{arrow bold}=[very thick,arrs]
\tikzstyle{arrow hide}=[draw=white!0,-]
\tikzstyle{arrow reverse}=[latex-]
\tikzstyle{cdnode}=[]

\tikzstyle{gn}=[dot,fill=lime!50,minimum width=0.2cm,inner sep=0.5pt,font=\footnotesize]
\tikzstyle{rn}=[dot,fill=red!50,inner sep=0.5pt,minimum width=0.2cm,font=\footnotesize]
\tikzstyle{bn}=[dot,fill=blue,minimum width=0.3cm]

\tikzstyle{rc}=[dot,thick,fill=white,draw = red,minimum width=0.2cm,inner sep=0.5pt,font=\footnotesize]
\tikzstyle{gc}=[dot,thick,fill=white,draw= lime,inner sep=0.5pt,minimum width=0.2cm,font=\footnotesize]
\tikzstyle{bc}=[dot,thick,fill=white,draw= blue,minimum width=0.3cm]

\tikzstyle{label}=[circle,fill=white,minimum width=0.3cm]

\tikzstyle{H box}=[rectangle,fill=yellow,draw=black,xscale=1,yscale=1,font=\small,inner sep=0.75pt,minimum width=0.15cm,minimum height=0.15cm]

\tikzstyle{clocklabel}=[dot,fill=yellow,draw=black,font=\tiny,inner sep=0.75pt]

\tikzstyle{rsn}=[circle split,draw,fill=red,font=\tiny,inner sep=0.75pt]
\tikzstyle{gsn}=[circle split,draw,fill=lime,font=\tiny,inner sep=0.75pt]
\tikzstyle{bsn}=[circle split,draw,fill=blue,font=\tiny,inner sep=0.75pt]

\tikzstyle{rsc}=[circle split,thick,draw= red,draw,fill=white,font=\tiny,inner sep=0.75pt]
\tikzstyle{gsc}=[circle split,thick,draw= lime,draw,fill=white,font=\tiny,inner sep=0.75pt]
\tikzstyle{bsc}=[circle split,thick,draw= blue,draw,fill=white,font=\tiny,inner sep=0.75pt]

\tikzstyle{cnot}=[fill=white,shape=circle,inner sep=-1.4pt]
\tikzstyle{wire label}=[font=\tiny, auto]

\newcommand{\denoteb}[1]{
\left\llbracket #1 \right\rrbracket}

\newcommand{\denote}[1]{
\llbracket #1 \rrbracket} 

\newcommand{\dempty}{%
\beginpgfgraphicnamed{scalars-s//emptysquare-small}
\InputIfFileExists{scalars-s//emptysquare-small.tikz}{}{\input{./figures/scalars-s//emptysquare-small.tikz}}
\endpgfgraphicnamed}

\newcommand{\drcup}{\raisebox{-0.15cm}{%
\beginpgfgraphicnamed{scalars-s//cup}
\InputIfFileExists{scalars-s//cup.tikz}{}{\input{./figures/scalars-s//cup.tikz}}
\endpgfgraphicnamed}}

\newcommand{\dsigmaf}{%
\beginpgfgraphicnamed{scalars-s//swapf}
\InputIfFileExists{scalars-s//swapf.tikz}{}{\input{./figures/scalars-s//swapf.tikz}}
\endpgfgraphicnamed}
\newcommand{\dHf}{%
\beginpgfgraphicnamed{scalars-s//Had4f}
\begin{tikzpicture}
	\begin{pgfonlayer}{nodelayer}
		\node [style={H box}] (0) at (0.5, 0) {};
		\node [style=none] (1) at (0.5, 0.3) {};
		\node [style=none] (2) at (0.5, -0.3) {};
	\end{pgfonlayer}
	\begin{pgfonlayer}{edgelayer}
		\draw (1.center) to (0);
		\draw (2.center) to (0);
	\end{pgfonlayer}
\end{tikzpicture}}
\endpgfgraphicnamed}
\newcommand{\drcupf}{\raisebox{-0.15cm}{%
\beginpgfgraphicnamed{scalars-s//cupf}
\InputIfFileExists{scalars-s//cupf.tikz}{}{\input{./figures/scalars-s//cupf.tikz}}
\endpgfgraphicnamed}}
\newcommand{\drcapf}{\raisebox{0.15cm}{%
\beginpgfgraphicnamed{scalars-s//capf}
\InputIfFileExists{scalars-s//capf.tikz}{}{\input{./figures/scalars-s//capf.tikz}}
\endpgfgraphicnamed}}

}
}

\newcommand{\did}{%
\beginpgfgraphicnamed{scalars-s//Id}
\begin{tikzpicture}
	\begin{pgfonlayer}{nodelayer}
		\node [style=none] (1) at (0.5, 0.3) {};
		\node [style=none] (2) at (0.5, -0.3) {};
	\end{pgfonlayer}
	\begin{pgfonlayer}{edgelayer}
		\draw (1.center) to (2.center);
	\end{pgfonlayer}
\end{tikzpicture}}
\endpgfgraphicnamed}
\newcommand{\didf}{%
\beginpgfgraphicnamed{scalars-s//Idf}
\begin{tikzpicture}
	\begin{pgfonlayer}{nodelayer}
		\node [style=none] (1) at (0.5, 0.3) {};
		\node [style=none] (2) at (0.5, -0.3) {};
	\end{pgfonlayer}
	\begin{pgfonlayer}{edgelayer}
		\draw (1.center) to (2.center);
	\end{pgfonlayer}
\end{tikzpicture}}
\endpgfgraphicnamed}

\newcommand{\gdzz}{%
\beginpgfgraphicnamed{scalars-s//RZ00alpha}
\begin{tikzpicture}
	\begin{pgfonlayer}{nodelayer}
		\node [style=gn] (0) at (0.5, 0) {\footnotesize$\alpha$};
	\end{pgfonlayer}
\end{tikzpicture}}
\endpgfgraphicnamed}
\newcommand{\gdzo}{%
\beginpgfgraphicnamed{scalars-s//RZ01alpha}
\begin{tikzpicture}
	\begin{pgfonlayer}{nodelayer}
		\node [style=gn] (0) at (0.5, 0.1) {\footnotesize$\alpha$};
		\node [style=none] (1) at (0.5, -0.2) {};
	\end{pgfonlayer}
	\begin{pgfonlayer}{edgelayer}
		\draw (1.center) to (0.center);
	\end{pgfonlayer}

\end{tikzpicture}}
\endpgfgraphicnamed}

\endpgfgraphicnamed}

\endpgfgraphicnamed}
\newcommand{\gdoo}{%
\beginpgfgraphicnamed{scalars-s//RZ11alpha}
\begin{tikzpicture}
	\begin{pgfonlayer}{nodelayer}
		\node [style=gn] (0) at (0.5, 0) {\footnotesize$\alpha$};
		\node [style=none] (1) at (0.5, -0.3) {};
		\node [style=none] (2) at (0.5, 0.3) {};
	\end{pgfonlayer}
	\begin{pgfonlayer}{edgelayer}
		\draw (1.center) to (0.center);
				\draw (2.center) to (0.center);
	\end{pgfonlayer}

\end{tikzpicture}}
\endpgfgraphicnamed}
\newcommand{\gdto}{%
\beginpgfgraphicnamed{scalars-s//RZ21alpha}
\InputIfFileExists{scalars-s//RZ21alpha.tikz}{}{\input{./figures/scalars-s//RZ21alpha.tikz}}
\endpgfgraphicnamed}
\newcommand{\gpi}{%
\beginpgfgraphicnamed{scalars-s//RZ00pi}
\begin{tikzpicture}
	\begin{pgfonlayer}{nodelayer}
		\node [style=gn] (0) at (0.5, 0) {\footnotesize$\pi$};
	\end{pgfonlayer}
\end{tikzpicture}}
\endpgfgraphicnamed}

\endpgfgraphicnamed}

\newcommand{\gdot}{%
\beginpgfgraphicnamed{scalars-s//RZ00zero}
\begin{tikzpicture}
	\begin{pgfonlayer}{nodelayer}
		\node [style=gn] (0) at (0.5, 0) {};
	\end{pgfonlayer}
\end{tikzpicture}
}
\endpgfgraphicnamed}

\newcommand{\rdzz}{%
\beginpgfgraphicnamed{scalars-s//RX00alpha}
\begin{tikzpicture}
	\begin{pgfonlayer}{nodelayer}
		\node [style=rn] (0) at (0.5, 0) {\footnotesize$\alpha$};
	\end{pgfonlayer}
\end{tikzpicture}}
\endpgfgraphicnamed}
\newcommand{\rdzo}{%
\beginpgfgraphicnamed{scalars-s//RX01alpha}
\begin{tikzpicture}
	\begin{pgfonlayer}{nodelayer}
		\node [style=rn] (0) at (0.5, 0.1) {\footnotesize$\alpha$};
		\node [style=none] (1) at (0.5, -0.2) {};
	\end{pgfonlayer}
	\begin{pgfonlayer}{edgelayer}
		\draw (1.center) to (0.center);
	\end{pgfonlayer}

\end{tikzpicture}}
\endpgfgraphicnamed}
\newcommand{\rdoo}{%
\beginpgfgraphicnamed{scalars-s//RX11alpha}
\begin{tikzpicture}
	\begin{pgfonlayer}{nodelayer}
		\node [style=rn] (0) at (0.5, 0) {\footnotesize$\alpha$};
		\node [style=none] (1) at (0.5, -0.3) {};
		\node [style=none] (2) at (0.5, 0.3) {};
	\end{pgfonlayer}
	\begin{pgfonlayer}{edgelayer}
		\draw (1.center) to (0.center);
				\draw (2.center) to (0.center);
	\end{pgfonlayer}

\end{tikzpicture}}
\endpgfgraphicnamed}

\tikzstyle{cdiag}=[matrix of math nodes, row sep=3em, column sep=3em, text height=1.5ex, text depth=0.25ex,inner sep=0.5em]
\tikzstyle{arrow above}=[transform canvas={yshift=0.5ex}]
\tikzstyle{arrow below}=[transform canvas={yshift=-0.5ex}]

\newcommand{\zxs}{\textup{ZX}_{\textsf{simp}}}

\newcommand{\zxc}{\textup{ZX}_{\textsf{ccc}}}
\newcommand{\zxb}{\textup{ZX}_{\textsf{bac}}}

\newcommand{\ZX}{\textsc{zx}}
\newcommand{\intf}[1]{\ensuremath{\left\llbracket #1 \right\rrbracket}} 
\tikzstyle{bscalar}=[star,fill=black,draw=black,minimum size=8pt,inner sep=0pt]

\tikzstyle{emptynode}=[circle,fill=none,draw=none,inner sep=1.5pt]

\usepackage{mathtools}

\makeatletter
\tikzset{%
  show node name/.code={%
    \expandafter\def\expandafter\tikz@mode\expandafter{\tikz@mode\show\tikz@fig@name}%
    }
}
\makeatother

\begin{document}

\title{Towards a Minimal Stabilizer ZX-calculus}
\author[M.~Backens]{Miriam Backens}
\address{School of Computer Science, University of Birmingham, UK}
\email{m.backens@cs.bham.ac.uk}

\author[S.~Perdrix]{Simon Perdrix}
\address{CNRS, LORIA, Inria-Mocqua; Universit\'e de Lorraine, France}
\email{simon.perdrix@loria.fr}

\author[Q.~Wang]{Quanlong Wang}
\address{Department of Computer Science, University of Oxford;  
Cambridge Quantum Computing Ltd., UK}
\email{quanlong.wang@cs.ox.ac.uk}

\begin{abstract}
The stabilizer \ZX-calculus is a rigorous graphical language for reasoning about quantum mechanics. The language is sound and complete: one can transform a stabilizer \ZX-diagram into another one using the graphical rewrite rules if and only if these two diagrams represent the same quantum evolution or quantum state. We previously showed that the stabilizer \ZX-calculus can be simplified by reducing the number of rewrite rules, without losing the property of completeness [Backens, Perdrix \& Wang, \textit{EPTCS}~236:1--20, 2017]. Here, we show that most of the remaining rules of the language are indeed necessary. We do however leave as an open question the necessity of two rules. These include, surprisingly, the bialgebra rule, which is an axiomatisation of complementarity, the cornerstone of the \ZX-calculus. Furthermore, we show that a weaker ambient category -- a braided autonomous category instead of the usual compact closed category -- is sufficient to recover the meta rule `only connectivity matters', even without assuming any symmetries of the generators.
\end{abstract}

\maketitle

\section{Introduction}

The \ZX-calculus is a high-level and intuitive graphical language for pure qubit quantum mechanics (QM), based on category theory \cite{coecke_interacting_2011}. It comes with a set of rewrite rules that potentially allow this graphical calculus to be used to replace matrix-based formalisms entirely for certain classes of problems. However, this replacement is only possible without losing deductive power if the \ZX-calculus is \emph{complete} for this class of problems, i.e.\ if any equality that is derivable using matrices can also be derived graphically.

The first fragment of the \ZX-calculus shown to be complete was the \emph{stabilizer \ZX-calculus} \cite{backens_zx-calculus_2013}. This fragment consists of the \ZX-diagrams involving angles which are multiples of $\pi/2$ only. The fragment of quantum theory that can be represented by stabilizer \ZX-diagrams is the so-called stabilizer quantum mechanics \cite{gottesman_stabilizer_1997}. Stabilizer QM is a non trivial fragment of quantum mechanics which is in fact efficiently classically simulatable \cite{gottesman_heisenberg_1998}   but which nevertheless exhibits
many important quantum properties, like entanglement and non-locality. It is furthermore of central importance in areas such as quantum error correcting codes \cite{nielsen_quantum_2010} and measurement-based quantum computation \cite{raussendorf_one-way_2001}.

A subset of these rules is also complete for the single-qubit Clifford+T group \cite{backens_zx-calculus_2014}.
Other fragments of the \ZX-calculus have recently been completed, these include the full Clifford+T fragment \cite{jeandel_complete_2017} as well as the full \ZX-calculus \cite{HNW,JPV-universal,ZXNormalForm,euler-zx}. The language can also be extended to capture mixed-state quantum mechanics \cite{carette2019completeness}. Nevertheless, we focus here on the stabilizer \ZX-calculus because it is the core of the overall language: all the fundamental structures -- e.g.\ the axiomatisation of complementary bases \cite{coecke_interacting_2011} -- are present in this fragment. The rule sets for larger parts of the formalism include the rules of the stabilizer \ZX-calculus with only minor modifications.

Now that the question of completeness has been resolved, we turn our attention to simplifying the \ZX-calculus, removing unnecessary equations while keeping only the essential axioms. This process simplifies the development, and potentially the efficiency, of automated tools for quantum reasoning, e.g.\ Quantomatic \cite{quanto}.

In a preliminary version of this work \cite{BPW16}, we gave a set of axioms that is significantly smaller than the usual one, containing just nine explicit rewrite rules. Previous rule sets usually contained about a dozen explicit rules and used the convention that any rule also holds with the colours red and green swapped or with the diagrams flipped upside-down, effectively nearly quadrupling the available set of rewrite rules.\footnote{Some rules are symmetric under the operations of swapping the colours and/or flipping them upside-down, hence the effective rule set is not quite four times the size of the explicitly-given one.} We showed that the colour symmetric and upside-down versions of the remaining rewrite rules can in fact be derived, so the convention is no longer required.

Here, we extend this work by showing that most of the remaining rules are indeed necessary, i.e.\ they cannot be derived from the other rules. Yet for two rules, the question of their necessity remains open; this includes the bialgebra rule which formalises the notion of complementary bases and thus plays core role in the language.

Furthermore, we consider the `only the connectivity matters' rule, which means that two diagrams represent the same matrix whenever one can be transformed into the other by moving components around without changing their connections. This meta-rule is an essential property of quantum diagrammatic reasoning, and refines the axioms of the ambient compact closed category. Indeed, the axioms of a compact closed category guarantee that two isomorphic diagrams are equivalent \cite{selinger_survey_2010}. The `only the connectivity matters' meta-rule  implies additionally that any two inputs or outputs of a generator can be freely exchanged.
We show that a single additional explicit rewrite rule is  sufficient to derive the symmetries of the generators, and thus the meta-rule `only the connectivity matters', from the simplified stabilizer \ZX-calculus together with the axioms of the ambient compact closed category (Section \ref{s:compact_closed}). More surprisingly, we show that a weaker ambient category is enough, namely a braided autonomous category (Section \ref{s:braided_category}). Graphically, this means that  3-dimensional isotopy is enough to derive the `only the connectivity matters' meta-rule.  

A preliminary version of this work has been published in the proceedings of the QPL'16 conference \cite{BPW16}. Soundness and completeness of the simplified \ZX-calculus are proved in \cite{BPW16}, together with the minimality of the scalar axioms (IV$'$) and (ZO$'$). In the present extended version, we prove the necessity of (almost) all the other rules of the language (section \ref{necessity}), and we also consider the simplification of the ambient category (section \ref{s:simplified_category}).

\section{A Simplified Stabilizer \ZX-calculus}\label{sec:zx}

The \ZX-calculus is a graphical language based on categorical quantum mechanics. This graphical notation is made rigorous by the underlying category theory \cite{coecke_interacting_2011,selinger_dagger_2007}.
For a less category-theoretical introduction to graphical languages of this type, see \cite{coecke_picturing_2017}.
We will examine the underlying category theory in more detail in Section~\ref{s:simplified_category}.

In this paper, we focus on the stabilizer fragment of the \ZX-calculus, as that encompasses many important aspects of the full language while also being complete. We introduce first the components of \ZX-diagrams and their interpretations, and then the rules of the language.

\subsection{Diagrams and standard interpretation}
\label{s:generators}

A diagram $D:k\to l$ of the stabilizer \ZX-calculus with $k$ inputs and $l$ outputs is generated by:
\begin{center}
\begin{tabular}{|r@{~}r@{~}c@{~}c|r@{~}r@{~}c@{~}c|}
\hline
$R_Z^{(n,m)}(\alpha)$&$:$&$n\to m$ & %
\beginpgfgraphicnamed{scalars-s//spideralpha}
\InputIfFileExists{scalars-s//spideralpha.tikz}{}{\input{./figures/scalars-s//spideralpha.tikz}}
\endpgfgraphicnamed & $R_X^{(n,m)}(\alpha)$&$:$&$ n\to m$& %
\beginpgfgraphicnamed{scalars-s//spiderredalpha}
\InputIfFileExists{scalars-s//spiderredalpha.tikz}{}{\input{./figures/scalars-s//spiderredalpha.tikz}}
\endpgfgraphicnamed\\
\hline
$H$&$:$&$1\to 1$ &%
\beginpgfgraphicnamed{scalars-s//Had4}
\InputIfFileExists{scalars-s//Had4.tikz}{}{\input{./figures/scalars-s//Had4.tikz}}
\endpgfgraphicnamed
 & $e $&$:$&$0 \to 0$ &%
\beginpgfgraphicnamed{scalars-s//emptysquare-small}
\InputIfFileExists{scalars-s//emptysquare-small.tikz}{}{\input{./figures/scalars-s//emptysquare-small.tikz}}
\endpgfgraphicnamed \\\hline
  $\sigma$&$:$&$ 2\to 2$& %
\beginpgfgraphicnamed{scalars-s//swap}
\InputIfFileExists{scalars-s//swap.tikz}{}{\input{./figures/scalars-s//swap.tikz}}
\endpgfgraphicnamed &$\mathbb I$&$:$&$1\to 1$&%
\beginpgfgraphicnamed{scalars-s//Id}
}
\endpgfgraphicnamed \\\hline
  $\epsilon$&$:$&$2\to 0$& %
\beginpgfgraphicnamed{scalars-s//cup}
\InputIfFileExists{scalars-s//cup.tikz}{}{\input{./figures/scalars-s//cup.tikz}}
\endpgfgraphicnamed &$\eta$&$:$&$ 0\to 2$&  %
\beginpgfgraphicnamed{scalars-s//cap}
\InputIfFileExists{scalars-s//cap.tikz}{}{\input{./figures/scalars-s//cap.tikz}}
\endpgfgraphicnamed
  \\\hline
\end{tabular}
\end{center}
where $m,n\in \mathbb N$, $\alpha \in \{\frac{k\pi}{2} | k\in\mathbb{Z}\}$, and $e$ is denoted by an empty diagram. Because of their many `legs', red and green dots are often called `spiders'.

When equal to $0$, the phase angles of the green and red dots may be omitted:
\[
\beginpgfgraphicnamed{scalars-s//spiderg}
\InputIfFileExists{scalars-s//spiderg.tikz}{}{\input{./figures/scalars-s//spiderg.tikz}}
\endpgfgraphicnamed := %
\beginpgfgraphicnamed{scalars-s//spidergz}
\InputIfFileExists{scalars-s//spidergz.tikz}{}{\input{./figures/scalars-s//spidergz.tikz}}
\endpgfgraphicnamed\qquad\qquad%
\beginpgfgraphicnamed{scalars-s//spiderr}
\InputIfFileExists{scalars-s//spiderr.tikz}{}{\input{./figures/scalars-s//spiderr.tikz}}
\endpgfgraphicnamed := %
\beginpgfgraphicnamed{scalars-s//spiderrz}
\InputIfFileExists{scalars-s//spiderrz.tikz}{}{\input{./figures/scalars-s//spiderrz.tikz}}
\endpgfgraphicnamed
\]

These components can be combined using the following two operations:
\begin{itemize}
\item Spacial composition: for any $D_1:a\to b$ and $D_2: c\to d$, $D_1\otimes D_2 : a+c\to b+d$ is constructed by placing $D_1$ and $D_2$ side-by-side, $D_2$ to the right of $D_1$.
\item Sequential composition: for any $D_1:a\to b$ and $D_2: b\to c$, $D_2\circ D_1 : a\to c$ is constructed by placing $D_1$ above $D_2$, connecting the outputs of $D_1$ to the inputs of $D_2$.
\end{itemize}

Spatial and sequential compositions satisfy that for any $D_1:a\to b$, $D_2:c\to d$, $D_3:b\to f$, and $D_4:d\to g$, $(D_3\otimes D_4)\circ (D_1\otimes D_2) = (D_3\circ D_1)\otimes (D_4\circ D_2)$. In other words, the \ZX-diagrams form a strict monoidal category which has natural numbers as objects: a diagram with $n$ inputs and $m$ outputs is a morphism $n\to m$, and the identity is inductively defined as $1_0 = e$ and $1_{1+n} = \mathbb I  \otimes 1_n$.   This property ensures that the standard interpretation of \ZX-diagrams, which we will now introduce, is well-defined: different ways of decomposing the same diagram in order to interpret it all yield the same interpretation \cite{coecke_interacting_2011}.

The standard interpretation associates with any \ZX-diagram $D:n\to m$ a linear map $\denoteb{D}:\mathbb C^{2^n}\to \mathbb C^{2^m}$, where $\mathbb{C}$ denotes the complex numbers. The interpretation is inductively defined as follows:
\[
 \begin{array}{rclcrclcrcl}
  \denoteb{D_1\otimes D_2}&:=&\denoteb{D_1}\otimes \denoteb{D_2} &\qquad&
  \denoteb{~\dHf~}&:=&\frac{1}{\sqrt{2}}\left(\begin{array}{@{}c@{}r@{}}1&1\\1&{~\text{-}}1\end{array}\right) &\qquad&
  \denoteb{~\dempty~}&:=&1 \\
  \denoteb{D_2\circ D_1}&:=&\denoteb{D_2}\circ\denoteb{D_1} &&
 &
 &&&
  \denote{\drcupf}&:=&\left(\begin{array}{@{}c@{}c@{}@{}c@{}c@{}}1&0&0&1\end{array}\right) \\
  \denoteb{~\dsigmaf~}&:=&\left(\begin{array}{@{}c@{}r@{}@{}c@{}c@{}}1&0&0&0\\0&0&1&0\\0&1&0&0\\0&0&0&1\\\end{array}\right) &&
  \denoteb{~\did~}&:=&\left(\begin{array}{@{}c@{}c@{}}1&0\\0&1\end{array}\right) &&
  \denoteb{\drcapf}&:=&\left(\begin{array}{@{}c@{}}1\\0\\0\\1\end{array}\right)
 \end{array}
\]

For green dots, $\denote{R_Z^{(0,0)}(\alpha)}:= 1{+}e^{i\alpha}$, and when $a{+}b>0$, $\denote{R_Z^{(a,b)}(\alpha)}$ is a matrix with $2^a$ columns and $2^b$ rows such that all entries are $0$ except the top left one which is $1$ and the bottom right one which is $e^{i\alpha}$, e.g.:
\[
 \def\arraystretch{0.5}
\denoteb{~\gdzz~} = 1+e^{i\alpha} \qquad \denoteb{~\gdzo~} = \left(\begin{array}{@{}c@{}}1\\e^{i\alpha}\end{array}\right) \qquad \denoteb{~\gdoo~} = \left(\begin{array}{@{}c@{}c@{}}1&~0~\\0&~e^{i\alpha}\end{array}\right) \qquad  \denoteb{\gdto} = \left(\begin{array}{@{}c@{}c@{}c@{}c@{}}1~&~0~&~0~&~0~\\0~&~0~&~0~&~e^{i\alpha}\end{array}\right)
\]

\noindent For any $a,b\ge  0$, $\denote{R_X^{a,b}(\alpha)}:= \denote{H}^{\otimes b}\circ \denote{R_Z^{a,b}(\alpha)} \circ \denote{H}^{\otimes a}$, where $M^{\otimes 0} =1$ and for any $k>0$, $M^{\otimes k}=M\otimes M^{\otimes k-1}$. E.g.,
\[
 \def\arraystretch{0.5}
 \denoteb{~\rdzz~} = 1+e^{i\alpha} \qquad \denoteb{~\rdzo~}= \sqrt2 e^{i\frac \alpha 2}\!\left(\begin{array}{@{}r@{}}\cos(\nicefrac{\alpha}2)\\\text{-}i\sin(\nicefrac{\alpha}2)\end{array}\right) \qquad \denoteb{~\rdoo~} = e^{i\frac \alpha 2}\! \left(\begin{array}{@{}r@{}r@{}}\cos(\nicefrac{\alpha}2)&\text{~~-}i\sin(\nicefrac{\alpha}2)\\\text{-}i\sin(\nicefrac{\alpha}2)&\cos(\nicefrac{\alpha}2)\end{array}\right)
\]

For a more involved example, consider the following diagram:
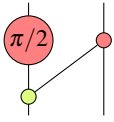
\[
\beginpgfgraphicnamed{scalars//interpretation_example}
\InputIfFileExists{scalars//interpretation_example.tikz}{}{\input{./figures/scalars//interpretation_example.tikz}}
\endpgfgraphicnamed
\]
Its standard interpretation can be found as follows:
\begin{align*}
 \denoteb{%
\beginpgfgraphicnamed{scalars//interpretation_example}
\InputIfFileExists{scalars//interpretation_example.tikz}{}{\input{./figures/scalars//interpretation_example.tikz}}
\endpgfgraphicnamed} &= \left( \denoteb{%
\beginpgfgraphicnamed{scalars//RZ21zero}
\InputIfFileExists{scalars//RZ21zero.tikz}{}{\input{./figures/scalars//RZ21zero.tikz}}
\endpgfgraphicnamed} \otimes \denoteb{~%
\beginpgfgraphicnamed{scalars-s//Id}
}
\endpgfgraphicnamed~} \right) \circ \left( \denoteb{%
\beginpgfgraphicnamed{scalars//RX11pi_2}
\begin{tikzpicture}
	\begin{pgfonlayer}{nodelayer}
		\node [style=rn] (0) at (0.5, 0) {\footnotesize$\pi/2$};
		\node [style=none] (1) at (0.5, -0.5) {};
		\node [style=none] (2) at (0.5, 0.5) {};
	\end{pgfonlayer}
	\begin{pgfonlayer}{edgelayer}
		\draw (1.center) to (0.center);
				\draw (2.center) to (0.center);
	\end{pgfonlayer}

\end{tikzpicture}
}
\endpgfgraphicnamed} \otimes \denoteb{%
\beginpgfgraphicnamed{scalars//RX12zero}
\InputIfFileExists{scalars//RX12zero.tikz}{}{\input{./figures/scalars//RX12zero.tikz}}
\endpgfgraphicnamed} \right) \\
 &= \left( \begin{pmatrix} 1&0&0&0\\0&0&0&1 \end{pmatrix} \otimes \begin{pmatrix} 1&0\\0&1 \end{pmatrix} \right) \circ \left( e^{i\frac \pi 4}\! \left(\begin{array}{@{}r@{}r@{}}\cos(\nicefrac{\pi}4)&\text{~~-}i\sin(\nicefrac{\pi}4)\\\text{-}i\sin(\nicefrac{\pi}4)&\cos(\nicefrac{\pi}4)\end{array}\right) \otimes \frac{1}{\sqrt{2}} \begin{pmatrix}1&0\\0&1\\0&1\\1&0\end{pmatrix} \right) \\
 &= \frac{e^{i\frac \pi 4}}{2} \begin{pmatrix} 1&0&-i&0 \\ 0&1&0&-i \\ 0&-i&0&1 \\ -i&0&1&0 \end{pmatrix}
\end{align*}
The category-theoretical underpinnings of the language ensure that all the different decompositions of a diagram yield the same interpretation.

\begin{rem}
 Of the three kinds of generators -- $R_Z^{(n,m)}$, $R_X^{(n,m)}$, and $H$ -- one could be eliminated without losing any expressive power.
 We nevertheless keep all three kinds of generators here, both for reasons of tradition and because this makes reasoning simpler.
 This approach is not inconsistent with the notion of working towards a minimal version of the stabilizer \ZX-calculus: we are looking for a version of the calculus where all rewrite rules are provably necessary, rather than the version with the smallest possible number of rules.
\end{rem}

The linear maps that can be represented by stabilizer \ZX-diagrams correspond to the so-called stabilizer fragment of quantum mechanics \cite{gottesman_stabilizer_1997}, which is generated by state preparations and measurements in the computational basis together with the group of Clifford unitaries.
All Clifford unitaries arise as quantum circuits over the gates
\[
 S = \begin{pmatrix}1&0\\0&i\end{pmatrix}, \qquad H = \frac{1}{\sqrt{2}} \begin{pmatrix}1&1\\1&-1\end{pmatrix} \quad\text{and}\quad C_X = \begin{pmatrix}1&0&0&0\\0&1&0&0\\0&0&0&1\\0&0&1&0\end{pmatrix}.
\]

Note that \ZX-diagrams with arbitrary angles (no longer necessarily multiples of $\frac \pi 2$) are universal:  for any $m,n\ge 0$ and any linear map $M:\mathbb C^{2^n}\to \mathbb C^{2^m}$, there exists a diagram $D:n\to m$ such that $\denote{D} = M$ \cite{coecke_interacting_2011}. When restricted to angles that are multiples of $\pi/4$, \ZX-diagrams are approximately universal, i.e.\ any linear map can approximated to arbitrary accuracy by such a \ZX-diagram. In this paper, we focus on the core of the \ZX-calculus formed by the stabilizer \ZX-diagrams.

\subsection{The rewrite rules, soundness and completeness}

The \ZX-calculus is not just a notation: it comes with a set of rewrite rules that allow equalities to be derived entirely graphically.
We are considering the stabilizer \ZX-calculus here because it is the fragment with the smallest complete set of rewrite rules. \emph{Complete} here means that any equality that can be derived using matrices can also be derived graphically using that set of rewrite rules \cite{backens_zx-calculus_2013,backens_making_2015}.

In addition to those explicit rewrite rules there is also a meta-rule: \emph{`only connectivity matters'} (previously stated as `only topology matters') \cite[Section~2.2.1]{coecke_interacting_2011},
which means that two diagrams represent the same matrix whenever one diagram can be transformed into the other by moving components around without changing their connections.
To formalise this, we first define a labelled graph associated with any \ZX-diagram.

\begin{defi}\label{dfn:labelled_graph}
 Define a set of labels
 \[
  \textstyle L:=\left\{ R_Z\left(\frac{k\pi}{2}\right) \,\middle|\, k\in\mathbb{Z}\right\} \cup \left\{ R_X\left(\frac{k\pi}{2}\right) \,\middle|\, k\in\mathbb{Z}\right\} \cup \left\{H \right\} \cup \{I_n\mid n\in\mathbb{N}_{\geq 1}\} \cup \{O_n\mid n\in\mathbb{N}_{\geq 1}\}.
 \]
 Given a \ZX-calculus diagram $D$, let $G_D=(V,E,\ell)$ be the labelled multigraph with vertices $V$, edges $E$, and labelling $\ell$ that arises as follows.
 The multigraph $(V,E)$ consists of:
 \begin{itemize}
  \item one vertex for each dot, Hadamard, input, or output of the diagram, and
  \item one edge for each edge in the original diagram, connecting the vertices corresponding to the endpoints of the original edge.
 \end{itemize}
 The labelling $\ell:V\to L$ is defined as follows:
 \begin{itemize}
  \item each vertex corresponding to a green dot with phase $\alpha$ is labelled $R_Z(\alpha)$,
  \item each vertex corresponding to a red dot with phase $\alpha$ is labelled $R_X(\alpha)$,
  \item each vertex corresponding to a Hadamard is labelled H,
  \item each vertex corresponding to an input has a unique label of the form $I_n$, where $n$ is the index of the input when counting from left to right, and
  \item each vertex corresponding to an output has a unique label of the form $O_n$, where $n$ is the index of the output when counting from left to right. 
 \end{itemize}
\end{defi}

\begin{defi}\label{dfn:connectivity}
 The rule `only connectivity matters' formally means the following:
 Suppose $D_1$ and $D_2$ are two \ZX-calculus diagrams.
 Then the two diagrams are equal if there exists a graph isomorphism from $G_{D_1}=(V_1,E_1,\ell_1)$ to $G_{D_2}=(V_2,E_2,\ell_2)$ which respects the labelling, i.e.\ an invertible map $h: V_1\to V_2$ such that
 \begin{itemize}
  \item if vertices $u,v\in V_1$ are connected by $n$ edges, then $h(u),h(v)$ are connected by $n$ edges, and
  \item for any $v\in V_1$, $\ell_1(v)=\ell_2(h(v))$.
 \end{itemize}
\end{defi}

 Definition~\ref{dfn:connectivity} implies that the category of \ZX-diagrams is symmetric:
\begin{lem}
With `only connectivity matters', the \ZX-diagrams form a symmetric monoidal category where for any $n,m\in \mathbb N$, $\sigma_{n,m}$ is the natural isomorphism inductively defined as: $\sigma_{0,0} := e$, $\sigma_{1,0}:=\mathbb I$, $\sigma_{1,1+m} := (1_1\otimes \sigma_{1,m})\circ(\sigma \otimes 1_m)$ and $\sigma_{2+n,m}:=(\sigma_{1,m}\otimes 1_{1+n})\circ (1_1\otimes \sigma_{1+n,m})$. 
\end{lem}

\begin{proof}
The `only connectivity matters' rule implies that $\sigma_{m,n}\circ \sigma_{n,m} = 1_{n+m}$ and for any $f:n\to n'$ and $g:m\to m'$, $(g\otimes f)\circ\sigma_{n,m} = \sigma_{m',n'}\circ (f\otimes g)$. E.g.~when $n=m'=n'=1$, $m=2$, $f=H$ and $g = R_Z^{(2,1)}(\pi/2)$: 

 \centerline{%
\beginpgfgraphicnamed{scalars-s//swap3}
\InputIfFileExists{scalars-s//swap3.tikz}{}{\input{./figures/scalars-s//swap3.tikz}}
\endpgfgraphicnamed\qquad\qquad%
\beginpgfgraphicnamed{scalars-s//swap4}
\InputIfFileExists{scalars-s//swap4.tikz}{}{\input{./figures/scalars-s//swap4.tikz}}
\endpgfgraphicnamed}
\noindent The other cases are analogous.
\end{proof}

The  `only connectivity matters' rule  also implies that the category of \ZX-diagrams is compact closed:

\begin{lem}
With `only connectivity matters', the \ZX-diagrams form a compact closed category where for any $n\in \mathbb N$, 
$\epsilon_0 = \eta_0 = e$,  $\epsilon_{1+n} = (1\otimes \sigma_{1,2n})\circ (\epsilon\otimes \epsilon_n)$, and $\eta_{1+n} = (\eta\otimes \eta_n)\circ (1\otimes  \sigma_{1,2n})$. 
\end{lem}

\begin{proof}
The `only connectivity matters' rule implies that for any $n\in \mathbb N$, $(\epsilon_n, \eta_n)$ is a compact structure, e.g.~when $n=1$:

 \centerline{%
\beginpgfgraphicnamed{scalars//compactstructure_cap}
\InputIfFileExists{scalars//compactstructure_cap.tikz}{}{\input{./figures/scalars//compactstructure_cap.tikz}}
\endpgfgraphicnamed \qquad\quad %
\beginpgfgraphicnamed{scalars//compactstructure_cup}
\InputIfFileExists{scalars//compactstructure_cup.tikz}{}{\input{./figures/scalars//compactstructure_cup.tikz}}
\endpgfgraphicnamed \qquad\quad %
\beginpgfgraphicnamed{scalars//compactstructure_snake}
\InputIfFileExists{scalars//compactstructure_snake.tikz}{}{\input{./figures/scalars//compactstructure_snake.tikz}}
\endpgfgraphicnamed}
 
\noindent The other cases are analogous
\end{proof}

It is known that compact closed categories enjoy the following graphical characterisation:

\begin{thmC}[{\cite[Theorem 14]{selinger_survey_2010}}]\label{thm:compact-closed}
 A well-formed equation between morphisms in the language of compact closed categories follows from the axioms of compact closed categories if and only if it holds, up to isomorphism of diagrams, in the graphical language.
\end{thmC}

The isomorphism of diagrams in the above theorem differs from the graph isomorphism in Definition~\ref{dfn:connectivity}: the isomorphism of diagrams induced by the axioms of a compact closed category preserves the order of inputs and outputs incident on each node in the diagram, analogous to the way the connectivity rule preserves the order of inputs and outputs of a diagram as a whole.

As a consequence, the `only connectivity matters' rule not only guarantees that the ambient category is compact closed, it also implies additional symmetry properties of the generators, e.g.:
\[
\beginpgfgraphicnamed{scalars-s//commute1}
\InputIfFileExists{scalars-s//commute1.tikz}{}{\input{./figures/scalars-s//commute1.tikz}}
\endpgfgraphicnamed =~%
\beginpgfgraphicnamed{scalars-s//commute2}
\InputIfFileExists{scalars-s//commute2.tikz}{}{\input{./figures/scalars-s//commute2.tikz}}
\endpgfgraphicnamed \qquad\qquad\qquad
\beginpgfgraphicnamed{scalars-s//bendingnew}
\InputIfFileExists{scalars-s//bendingnew.tikz}{}{\input{./figures/scalars-s//bendingnew.tikz}}
\endpgfgraphicnamed ~=~ %
\beginpgfgraphicnamed{scalars-s//nonbending}
\InputIfFileExists{scalars-s//nonbending.tikz}{}{\input{./figures/scalars-s//nonbending.tikz}}
\endpgfgraphicnamed ~=~ %
\beginpgfgraphicnamed{scalars-s//bendingnew2}
\InputIfFileExists{scalars-s//bendingnew2.tikz}{}{\input{./figures/scalars-s//bendingnew2.tikz}}
\endpgfgraphicnamed
   \]

Note that instead of imposing `only connectivity matters' as a rule, one can derive it from the axioms of the ambient category together with some extra axioms, see Section~\ref{s:compact_closed}.

\begin{figure}
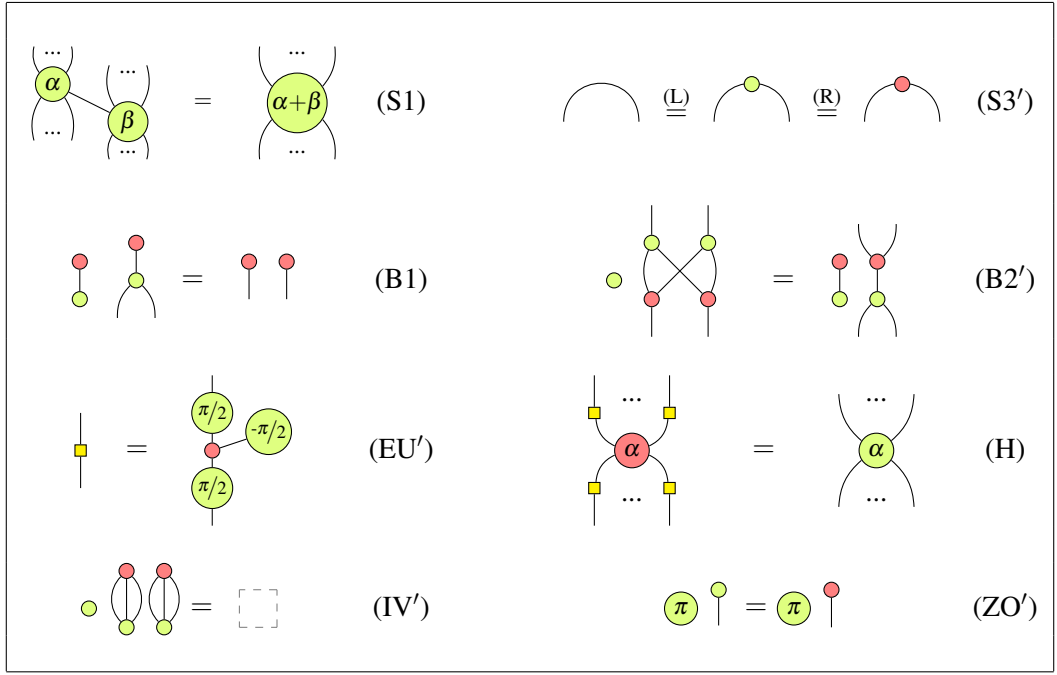

 \centering
 \begin{tabular}{|ccccc|}
  \hline
  &&&& \\
\beginpgfgraphicnamed{scalars//spider-bis}
\InputIfFileExists{scalars//spider-bis.tikz}{}{\input{./figures/scalars//spider-bis.tikz}}
\endpgfgraphicnamed&(S1)&$\qquad$&%
\beginpgfgraphicnamed{scalars//induced_compact_structure}
\InputIfFileExists{scalars//induced_compact_structure.tikz}{}{\input{./figures/scalars//induced_compact_structure.tikz}}
\endpgfgraphicnamed&(S3$'$)\\
  &&&& \\
\beginpgfgraphicnamed{scalars//b1s}
\InputIfFileExists{scalars//b1s.tikz}{}{\input{./figures/scalars//b1s.tikz}}
\endpgfgraphicnamed&(B1)&&%
\beginpgfgraphicnamed{scalars//b2snew}
\InputIfFileExists{scalars//b2snew.tikz}{}{\input{./figures/scalars//b2snew.tikz}}
\endpgfgraphicnamed&(B2$'$)\\
  &&&& \\
\beginpgfgraphicnamed{scalars//HadaDecomSingles-prime}
\InputIfFileExists{scalars//HadaDecomSingles-prime.tikz}{}{\input{./figures/scalars//HadaDecomSingles-prime.tikz}}
\endpgfgraphicnamed&(EU$'$)&&%
\beginpgfgraphicnamed{scalars//h2}
\InputIfFileExists{scalars//h2.tikz}{}{\input{./figures/scalars//h2.tikz}}
\endpgfgraphicnamed&(H)\\
  &&&& \\
\beginpgfgraphicnamed{scalars//dotinverse}
\InputIfFileExists{scalars//dotinverse.tikz}{}{\input{./figures/scalars//dotinverse.tikz}}
\endpgfgraphicnamed &(IV$'$)&&%
\beginpgfgraphicnamed{scalars//zo1-prime}
\InputIfFileExists{scalars//zo1-prime.tikz}{}{\input{./figures/scalars//zo1-prime.tikz}}
\endpgfgraphicnamed&(ZO$'$)\\
  &&&& \\
  \hline
 \end{tabular}
 \caption{Simplified rules for the stabilizer \ZX-calculus, using the conventions that the right-hand side of (IV$'$) is an empty diagram and that ellipsis denote zero or more wires. Note that the addition in (S1) is standard addition (not modular addition).}\label{figure1}
\end{figure}

We can now define the main set of \ZX-calculus rules employed and analysed in this paper.
Instead of the traditional set of rewrite rules for the stabilizer \ZX-calculus, we use a new, simpler, set of rules first introduced in \cite{BPW16}, which we denote by $\zxs$.

\begin{defi}
 The rule set $\zxs$ consists of the graphical rewrite rules given in Figure \ref{figure1} together with the metarule `only connectivity matters'.
\end{defi}

This set consists of 9 axioms, plus the `only connectivity matters' axiom described in Definition~\ref{dfn:connectivity}. The set of axioms of Figure \ref{figure1} is significantly simpler and more compact than previous versions of the stabilizer \ZX-calculus.

Many versions of the \ZX-calculus restrict the phase angles to lie within the range $(-\pi,\pi]$ or $[0,2\pi)$ and define the addition operation used in the rule (S1) to be addition modulo $2\pi$.
We show that this assumption is not in fact necessary.

For any pair of \ZX-diagrams $D_1,D_2$, and any set of rewrite rules $R$, we denote by $R\vdash D_1=D_2$ the statement that $D_1$ can be transformed into $D_2$ using the rules of $R$.

\begin{thm}\label{thm:2pi=0}
 $\zxs\vdash\; %
\beginpgfgraphicnamed{scalars/2pi-eq-zero}
\InputIfFileExists{scalars/2pi-eq-zero.tikz}{}{\input{./figures/scalars/2pi-eq-zero.tikz}}
\endpgfgraphicnamed$.
\end{thm}
\begin{proof}
 Note first that the derivation of the Hopf law
 \begin{equation}\label{eq:hopf}
\beginpgfgraphicnamed{scalars/hopf}
\InputIfFileExists{scalars/hopf.tikz}{}{\input{./figures/scalars/hopf.tikz}}
\endpgfgraphicnamed
 \end{equation}
 does not involve any phases (see e.g.\ \cite[Lemma~A.3]{BPW16}), so it is unaffected by dropping the restriction of phases to an interval of length $2\pi$.
 The same holds for the proofs of
 \begin{equation}\label{eq:inverse-orig}
\beginpgfgraphicnamed{scalars//inverserule}
\InputIfFileExists{scalars//inverserule.tikz}{}{\input{./figures/scalars//inverserule.tikz}}
\endpgfgraphicnamed
 \end{equation}
 and of
 \begin{equation}\label{eq:Had-selfinv}
\beginpgfgraphicnamed{scalars//H2-small}
\begin{tikzpicture}[scale=0.6]
	\begin{pgfonlayer}{nodelayer}
		\node [style=none] (1) at (0, 0.5) {};
		\node [style=none] (2) at (0, -0.5) {};
		\node [style=H box] (3) at (0, -0.25) {};
		\node [style=H box] (4) at (0, 0.25) {};
	\end{pgfonlayer}
	\begin{pgfonlayer}{edgelayer}
		\draw (2.center) to (1.center);
	\end{pgfonlayer}
\end{tikzpicture}
}
\endpgfgraphicnamed~=~ %
\beginpgfgraphicnamed{scalars//Id-2}
\begin{tikzpicture}[scale=0.6]
	\begin{pgfonlayer}{nodelayer}
		\node [style=none] (1) at (0, 0.5) {};
		\node [style=none] (2) at (0, -0.5) {};
	\end{pgfonlayer}
	\begin{pgfonlayer}{edgelayer}
		\draw (2.center) to (1.center);
	\end{pgfonlayer}
\end{tikzpicture}
}
\endpgfgraphicnamed
 \end{equation}
 see e.g.\ \cite[Lemma~3.2]{BPW16} and \cite[Lemma~A.1]{BPW16}.
 
 In the following, we will abbreviate phases $\pm\frac{\pi}{2}$ to the labels $\pm$.
 The `only connectivity matters' rule will be used implicitly where appropriate.
 We have
 \begin{equation}\label{eq:2pi-step1}
\beginpgfgraphicnamed{scalars/2pi-proof-1}
\InputIfFileExists{scalars/2pi-proof-1.tikz}{}{\input{./figures/scalars/2pi-proof-1.tikz}}
\endpgfgraphicnamed
 \end{equation}
 Furthermore, we can show
 \begin{equation}\label{eq:2pi-step2}
\beginpgfgraphicnamed{scalars/2pi-proof-2}
\InputIfFileExists{scalars/2pi-proof-2.tikz}{}{\input{./figures/scalars/2pi-proof-2.tikz}}
\endpgfgraphicnamed
 \end{equation}
 Combining these two results yields
 \begin{equation}\label{eq:2pi-step3}
\beginpgfgraphicnamed{scalars/2pi-proof-3}
\InputIfFileExists{scalars/2pi-proof-3.tikz}{}{\input{./figures/scalars/2pi-proof-3.tikz}}
\endpgfgraphicnamed
 \end{equation}
 Additionally, we have
 \begin{multline}\label{eq:2pi-step4}
\beginpgfgraphicnamed{scalars/2pi-proof-4}
\InputIfFileExists{scalars/2pi-proof-4.tikz}{}{\input{./figures/scalars/2pi-proof-4.tikz}}
\endpgfgraphicnamed \\
\beginpgfgraphicnamed{scalars/2pi-proof-5}
\InputIfFileExists{scalars/2pi-proof-5.tikz}{}{\input{./figures/scalars/2pi-proof-5.tikz}}
\endpgfgraphicnamed
 \end{multline}
 Thus we can show:
 \[
\beginpgfgraphicnamed{scalars/2pi-proof}
\InputIfFileExists{scalars/2pi-proof.tikz}{}{\input{./figures/scalars/2pi-proof.tikz}}
\endpgfgraphicnamed
 \]
 completing the proof.
\end{proof}

\begin{cor}\label{cor:2kpi=0}
 For any $k\in\mathbb{Z}$,
 \[
  \zxs\vdash \; %
\beginpgfgraphicnamed{scalars-s/spiderg-a-2pi}
\InputIfFileExists{scalars-s/spiderg-a-2pi.tikz}{}{\input{./figures/scalars-s/spiderg-a-2pi.tikz}}
\endpgfgraphicnamed\; = %
\beginpgfgraphicnamed{scalars-s/spiderg-a}
\InputIfFileExists{scalars-s/spiderg-a.tikz}{}{\input{./figures/scalars-s/spiderg-a.tikz}}
\endpgfgraphicnamed
 \]
\end{cor}
\begin{proof}
 The case $k=0$ is trivial.
 We demonstrate the case $k=1$:
 \begin{equation}\label{eq:a-2pi}
\beginpgfgraphicnamed{scalars-s/spiderg-a-2pi-proof}
\InputIfFileExists{scalars-s/spiderg-a-2pi-proof.tikz}{}{\input{./figures/scalars-s/spiderg-a-2pi-proof.tikz}}
\endpgfgraphicnamed
 \end{equation}
 For any positive $k$, repeated application of \eqref{eq:a-2pi} yields the desired result.
 If $k$ is negative, apply \eqref{eq:a-2pi} from right to left instead.
\end{proof}

\begin{thm}
 The simplified rule set is sound and complete, i.e. for any two \ZX-calculus diagrams $D_1$ and $D_2$, we have:
 \[
  \zxs\vdash D_1=D_2 \quad\Longleftrightarrow\quad \intf{D_1}=\intf{D_2}.
 \]
\end{thm}
\begin{proof}
 By Corollary~\ref{cor:2kpi=0}, $\zxs$ is equivalent to the presentation of the \ZX-calculus used in \cite{BPW16}.
 The result thus follows from \cite[Theorems~4.1 \& 4.2]{BPW16}.
\end{proof}

\section{On the necessity of the rewrite rules in the simplified set}
\label{necessity}

The set of rules $\zxs$ is sound and complete. Can it be further simplified? We show in the following series of lemmas, that 8 of the 9 explicit rewrite rules are actually necessary, i.e.~they cannot be derived from the other rules of the language.
Note that (S1) and (H) are actually infinite families of rules; we consider these necessary if at least one of the instantiations is necessary.

Indeed, we begin by considering (S1), a key rule of the \ZX-calculus which acts on the degree of  the dots, and gives to the parameters of the dots their group structure. This rule cannot be derived from the other rules of the \ZX-calculus:

\begin{lem}
The (S1) rule is necessary: $\zxs \setminus \textup{(S1)} \not \vdash  \textup{(S1)}$. 
\end{lem}

\begin{proof}
All rules but (S1) are sound with respect to the following interpretation: for any diagram $D$, let $\denote D^{(S1)}\in \mathbb R$ be inductively defined as $\denote {D_1\otimes D_2}^{(S1)} =\denote {D_1\circ  D_2}^{(S1)} = \denote {D_1}^{(S1)}\denote {D_2}^{(S1)}$,
 \[
  \intf{~%
\beginpgfgraphicnamed{scalars-s/spiderr-a}
\InputIfFileExists{scalars-s/spiderr-a.tikz}{}{\input{./figures/scalars-s/spiderr-a.tikz}}
\endpgfgraphicnamed~}^{(S1)}=\intf{~%
\beginpgfgraphicnamed{scalars-s/spiderg-a}
\InputIfFileExists{scalars-s/spiderg-a.tikz}{}{\input{./figures/scalars-s/spiderg-a.tikz}}
\endpgfgraphicnamed~}^{(S1)}=\frac{\sqrt 5-1}\pi \alpha+1,
 \]
 and $\denote{.}^{(S1)}=1$ for all the other generators.
 
 We now show that all rules except (S1) are sound under this interpretation. Indeed, all the angle-free rules are trivially sound. (ZO$'$) and (H) are also sound since the green and red dots have the same interpretation and the interpretation of $H$ is trivial. Moreover, (EU$'$) is sound, as
 \[
  \intf{%
\beginpgfgraphicnamed{scalars-s//eulernewanglepart}
\InputIfFileExists{scalars-s//eulernewanglepart.tikz}{}{\input{./figures/scalars-s//eulernewanglepart.tikz}}
\endpgfgraphicnamed}^{(S1)} = \left(\frac{\sqrt 5-1}\pi \frac \pi 2+1\right)^2 \left(\frac{\sqrt 5-1}\pi \left(\frac{-\pi} 2\right) + 1\right) = 1 = \intf{~\dHf~}.
 \]
 Notice that the connectivity meta-rule is also sound: the interpretation $\intf{.}^{(S1)}$ depends solely on the phases of dots, so it is an invariant of the isomorphism classes of labelled graphs in Definition~\ref{dfn:labelled_graph}.
 
 However, (S1) is not sound e.g. when $\alpha=\beta=\frac \pi2$:
 \[
  \intf{%
\beginpgfgraphicnamed{scalars//spider-bis-l-pi2}
\InputIfFileExists{scalars//spider-bis-l-pi2.tikz}{}{\input{./figures/scalars//spider-bis-l-pi2.tikz}}
\endpgfgraphicnamed}^{(S1)} = \left(\frac{\sqrt 5-1}\pi \frac \pi 2+1\right)^2 = \frac{\sqrt{5}+3}{2} \neq \sqrt{5} = \frac{\sqrt 5-1}\pi \pi+1 = \intf{%
\beginpgfgraphicnamed{scalars//spider-bis-r-pi}
\InputIfFileExists{scalars//spider-bis-r-pi.tikz}{}{\input{./figures/scalars//spider-bis-r-pi.tikz}}
\endpgfgraphicnamed}^{(S1)}
 \]
 As a consequence (S1) cannot be derived from the other rules. 
\end{proof}

To avoid issues of normalisation, the following two proofs employ interpretations of diagrams as relations instead of linear maps.
Analogous to the representation of linear maps as complex matrices, we represent relations as logical matrices with elements in $\mathbb{B}=\{0,1\}$.
Given a relation $R:A\to B$, the rows of the corresponding logical matrix $M_R$ are indexed by elements of $A$, the columns are indexed by elements of $B$, and the element $(M_R)_{ab}$ is 1 if and only if $(a,b)\in R$.
We denote by $M\otimes N$ the Kronecker product of the logical matrices $M$ and $N$, this corresponds to the Cartesian product of the underlying relations.
The matrix product of two logical matrices $M:\mathbb{B}^{r}\to\mathbb{B}^{s}$ and $N:\mathbb{B}^{s}\to\mathbb{B}^{t}$ is denoted by $N\circ M:\mathbb{B}^{r}\to\mathbb{B}^{t}$, and defined as $(N\circ M)_{jk} = \bigvee_{\ell=1}^{s} (N_{j\ell}\wedge M_{\ell k})$.
Here, $\vee$ denotes logical disjunction and $\wedge$ denotes logical conjunction.
This corresponds to the usual notion of relational composition.

In our interpretation, a diagram $D:n\to m$ will be associated with a logical matrix $\mathbb{B}^{2^n}\to\mathbb{B}^{2^m}$.
There are only two scalars in this model, 0 and 1, so we do not need to worry about normalisation or scaling of diagrams.

For both proofs, the structural maps and diagram compositions will be interpreted as follows:
\begin{align*}
 \intf{D_1\otimes D_2}^{rel} &:= \intf{D_1}^{rel}\otimes\intf{D_2}^{rel} & & &
 \intf{D_1\circ D_2}^{rel} &:= \intf{D_1}^{rel}\circ\intf{D_2}^{rel} \\
 \intf{~\dempty~}^{rel} &:= 1 &
 \intf{~\didf~}^{rel} &:= \begin{pmatrix}1&0\\0&1\end{pmatrix} &
 \intf{\drcupf}^{rel} &:= \begin{pmatrix}1&0&0&1\end{pmatrix} \\
 \intf{~\dsigmaf~}^{rel} &:= \begin{pmatrix}1&0&0&0\\0&0&1&0\\0&1&0&0\\0&0&0&1\end{pmatrix} & & &
 \intf{\drcapf}^{rel} &:= \begin{pmatrix}1\\0\\0\\1\end{pmatrix}
\end{align*}
Note the matrices have the same values as the matrices of the standard interpretation.
The tensor product of logical matrix is the same as the tensor product of standard matrices.
Furthermore, all the matrices above have the property of having at most a single 1 in each row or column.
It is straightforward to check that, for such matrices, the logical matrix product is the same as the standard matrix product.
Thus this choice of interpretation satisfies all the equations of a compact closed category.
In the following two lemmas, we will extend this interpretation to the nodes of \ZX-diagrams in different ways.

The (S3$'$L) rule guarantees that the `green' compact structure coincides with the compact structure of the ambient category, and  (S3$'$R) that the `green' and `red' compact structures coincide. While we do not know whether (S3$'$R) is necessary or not (see Lemma \ref{lemma:B2S3R}), the (S3$'$L) rule cannot be derived from the other rules of the language:

\begin{lem}\label{lem:disconnecting_functor}
The (S3$'$L) rule is necessary: $\zxs \setminus \textup{(S3$'$L)} \not \vdash  \textup{(S3$'$L)}$. 
\end{lem}

\begin{proof}
 To prove the necessity of (S3$'$L), which relates a wire to a spider, we employ an interpretation which completely `disconnects' any map not consisting solely of wires, the relational equivalent of interpreting every node as %
\beginpgfgraphicnamed{scalars-s//disconnect}
\InputIfFileExists{scalars-s//disconnect.tikz}{}{\input{./figures/scalars-s//disconnect.tikz}}
\endpgfgraphicnamed .

 For any diagram $D$, let the relation $\intf{D}^{(S3'L)}$ be the extension of $\intf{.}^{rel}$ which satisfies:
\begin{align*}
 \intf{~\dHf~}^{(S3'L)} &= \begin{pmatrix}1&1\\1&1\end{pmatrix}  &
 \intf{~%
\beginpgfgraphicnamed{scalars-s/spiderg-a}
\InputIfFileExists{scalars-s/spiderg-a.tikz}{}{\input{./figures/scalars-s/spiderg-a.tikz}}
\endpgfgraphicnamed~}^{(S3'L)} =
 \intf{~%
\beginpgfgraphicnamed{scalars-s/spiderr-a}
\InputIfFileExists{scalars-s/spiderr-a.tikz}{}{\input{./figures/scalars-s/spiderr-a.tikz}}
\endpgfgraphicnamed~}^{(S3'L)} &= \begin{pmatrix}1&1&\ldots&1\\ 1&1&\ldots&1\\ \vdots&\vdots&\ddots&\vdots\\ 1&1&\ldots&1 \end{pmatrix}.
\end{align*}
 This interpretation achieves the desired disconnection by having no correlation between the inputs and outputs of any spider, or of the Hadamard node, as opposed to the perfect correlation between the endpoints of a wire.

 Logical matrices of all-ones (trivially) satisfy the spider rule.
 Furthermore, all scalar diagrams are interpreted as 1; for example, $\intf{\gdzz}=\intf{\rdzz}=1$.
 Thus, for each of (S1), (S3$'$R) (B1), (B2$'$), (EU$'$), (H), (IV$'$), and (ZO$'$), both sides of the rule are interpreted as a logical matrix of all-ones, meaning these rules are all sound.
 For example, consider the rule (B2$'$).
 The LHS becomes
 \begin{align*}
  \intf{%
\beginpgfgraphicnamed{scalars//b2snew-l}
\InputIfFileExists{scalars//b2snew-l.tikz}{}{\input{./figures/scalars//b2snew-l.tikz}}
\endpgfgraphicnamed}^{(S3'L)} &= \intf{\gdot}^{(S3'L)} \otimes  \left( \left[ \intf{%
\beginpgfgraphicnamed{scalars-s//RX21zero}
\InputIfFileExists{scalars-s//RX21zero.tikz}{}{\input{./figures/scalars-s//RX21zero.tikz}}
\endpgfgraphicnamed}^{(S3'L)} \otimes \intf{%
\beginpgfgraphicnamed{scalars-s//RX21zero}
\InputIfFileExists{scalars-s//RX21zero.tikz}{}{\input{./figures/scalars-s//RX21zero.tikz}}
\endpgfgraphicnamed}^{(S3'L)} \right] \circ \right. \\
  &\qquad \left. \left[ \intf{~\didf~}^{(S3'L)} \otimes \intf{~\dsigmaf~}^{(S3'L)} \otimes \intf{~\didf~}^{(S3'L)} \right] \circ \left[ \intf{%
\beginpgfgraphicnamed{scalars-s//RZ12zero}
\InputIfFileExists{scalars-s//RZ12zero.tikz}{}{\input{./figures/scalars-s//RZ12zero.tikz}}
\endpgfgraphicnamed}^{(S3'L)} \otimes \intf{%
\beginpgfgraphicnamed{scalars-s//RZ12zero}
\InputIfFileExists{scalars-s//RZ12zero.tikz}{}{\input{./figures/scalars-s//RZ12zero.tikz}}
\endpgfgraphicnamed}^{(S3'L)} \right] \right) \\
  &= 1 \otimes \left( \left[ \begin{pmatrix}1&1&1&1\\1&1&1&1\end{pmatrix} \otimes \begin{pmatrix}1&1&1&1\\1&1&1&1\end{pmatrix} \right] \circ \right. \\
  &\qquad \left. \left[ \begin{pmatrix}1&0\\0&1\end{pmatrix} \otimes \begin{pmatrix}1&0&0&0\\0&0&1&0\\0&1&0&0\\0&0&0&1\end{pmatrix} \otimes \begin{pmatrix}1&0\\0&1\end{pmatrix} \right] \circ \left[ \begin{pmatrix} 1&1\\ 1&1\\ 1&1\\ 1&1 \end{pmatrix} \otimes \begin{pmatrix} 1&1\\ 1&1\\ 1&1\\ 1&1 \end{pmatrix} \right] \right)
  \intertext{Now, the middle brackets contain a permutation matrix while the other two pairs of brackets contain matrices of all-ones. Using the multiplication rule for logical matrices, the result is a matrix of all-ones}
  &= \begin{pmatrix}1&1&1&1\\1&1&1&1\\1&1&1&1\\1&1&1&1\end{pmatrix}.
  \intertext{Similarly, for the RHS we find}
  \intf{%
\beginpgfgraphicnamed{scalars//b2snew-r}
\InputIfFileExists{scalars//b2snew-r.tikz}{}{\input{./figures/scalars//b2snew-r.tikz}}
\endpgfgraphicnamed}^{(S3'L)} &= \left[ \intf{%
\beginpgfgraphicnamed{scalars-s//RZ10z}
\begin{tikzpicture}[yscale=-1]
	\begin{pgfonlayer}{nodelayer}
		\node [style=gn] (0) at (0.5, 0.1) {};
		\node [style=none] (1) at (0.5, -0.2) {};
	\end{pgfonlayer}
	\begin{pgfonlayer}{edgelayer}
		\draw (1.center) to (0.center);
	\end{pgfonlayer}

\end{tikzpicture}
}
\endpgfgraphicnamed}^{(S3'L)} \circ \intf{%
\beginpgfgraphicnamed{scalars-s//RX01z}
\begin{tikzpicture}
	\begin{pgfonlayer}{nodelayer}
		\node [style=rn] (0) at (0.5, 0.1) {};
		\node [style=none] (1) at (0.5, -0.2) {};
	\end{pgfonlayer}
	\begin{pgfonlayer}{edgelayer}
		\draw (1.center) to (0.center);
	\end{pgfonlayer}

\end{tikzpicture}
}
\endpgfgraphicnamed}^{(S3'L)} \right] \otimes \left[ \intf{%
\beginpgfgraphicnamed{scalars-s//RZ12zero}
\InputIfFileExists{scalars-s//RZ12zero.tikz}{}{\input{./figures/scalars-s//RZ12zero.tikz}}
\endpgfgraphicnamed}^{(S3'L)} \circ \intf{%
\beginpgfgraphicnamed{scalars-s//RX21zero}
\InputIfFileExists{scalars-s//RX21zero.tikz}{}{\input{./figures/scalars-s//RX21zero.tikz}}
\endpgfgraphicnamed}^{(S3'L)} \right] \\
 &= \left[ \begin{pmatrix}1&1\end{pmatrix} \circ \begin{pmatrix}1\\1\end{pmatrix} \right] \otimes \left[ \begin{pmatrix} 1&1\\ 1&1\\ 1&1\\ 1&1 \end{pmatrix} \circ \begin{pmatrix}1&1&1&1\\1&1&1&1\end{pmatrix} \right]
 = \begin{pmatrix}1&1&1&1\\1&1&1&1\\1&1&1&1\\1&1&1&1\end{pmatrix},
 \end{align*}
 so (B2$'$) is sound.
 Yet for (S3$'$L) we have
 \[
  \intf{~%
\beginpgfgraphicnamed{scalars//cap}
\begin{tikzpicture}
	\begin{pgfonlayer}{nodelayer}
		\node [style=none] (0) at (0, 0.15) {};
		\node [style=none] (1) at (-0.3, -0.15) {};
		\node [style=none] (2) at (0.3, -0.15) {};
	\end{pgfonlayer}
	\begin{pgfonlayer}{edgelayer}
		\draw [bend left=45, looseness=1.00] (1.center) to (0.center);
		\draw [bend right=45, looseness=1.00] (2.center) to (0.center);
	\end{pgfonlayer}
\end{tikzpicture}}
\endpgfgraphicnamed~}^{(S3'L)} = \begin{pmatrix}1&0&0&1\end{pmatrix}^T \neq \begin{pmatrix}1&1&1&1\end{pmatrix}^T = \intf{~%
\beginpgfgraphicnamed{scalars//green_cap}
\InputIfFileExists{scalars//green_cap.tikz}{}{\input{./figures/scalars//green_cap.tikz}}
\endpgfgraphicnamed~}^{(S3'L)}.
 \]
 Thus, (S3$'$L) cannot be derived from the other rules.
\end{proof}

The copy rule (B1), which states that `green copies red', is one of the fundamental rules of the \ZX-calculus. This rule is necessary:

\begin{lem}
 The copy rule (B1) is necessary:  $\zxs \setminus \textup{(B1)} \not \vdash  \textup{(B1)}$. 
\end{lem}
\begin{proof}
 The copy rule (B1) is the only rule which maps a diagram in which there exists a path between any pair of external legs to a diagram in which there does not exist such a path.
 To prove its necessity, we thus use an interpretation which emphasises connectivity, the relational equivalent of interpreting every node as a green spider
 \ctikzfig{scalars-s//spiderg}

 For any diagram $D$, let the relation $\intf{D}^{(B1)}$ be defined by the extension of $\intf{.}^{rel}$ which satisfies:
\begin{align*}
 \intf{~\dHf~}^{(B1)} &= \begin{pmatrix}1&0\\0&1\end{pmatrix} &
 \intf{~%
\beginpgfgraphicnamed{scalars-s/spiderg-a}
\InputIfFileExists{scalars-s/spiderg-a.tikz}{}{\input{./figures/scalars-s/spiderg-a.tikz}}
\endpgfgraphicnamed~}^{(B1)} =
 \intf{~%
\beginpgfgraphicnamed{scalars-s/spiderr-a}
\InputIfFileExists{scalars-s/spiderr-a.tikz}{}{\input{./figures/scalars-s/spiderr-a.tikz}}
\endpgfgraphicnamed~}^{(B1)} &= \begin{pmatrix} 1&0&\ldots&0\\ 0&0&\ldots&0\\ \vdots&\vdots&\ddots&\vdots\\ 0&0&\ldots&1 \end{pmatrix},
\end{align*}
 where only the top left and the bottom right element of each matrix are 1.
 Relations of this form satisfy the spider law, so any connected diagram is interpreted as a logical matrix in which exactly the top left and the bottom right element are 1.
 Hence, connectivity corresponds to perfect correlation between all inputs and outputs.

 Again, all scalar diagrams are interpreted as 1 -- for example, $\intf{\gdzz}=\intf{\rdzz}=1$ -- so any scaling of a connected diagram is interpreted as a logical matrix of the above form.
 This means (S1), (S3$'$), (B2$'$), (EU$'$), (H), (IV$'$) and (ZO$'$) are all sound.
 For example, for the rule (B2$'$), the LHS gives
 \begin{align*}
  \intf{%
\beginpgfgraphicnamed{scalars//b2snew-l}
\InputIfFileExists{scalars//b2snew-l.tikz}{}{\input{./figures/scalars//b2snew-l.tikz}}
\endpgfgraphicnamed}^{(S3'L)} &= \intf{\gdot}^{(S3'L)} \otimes  \left( \left[ \intf{%
\beginpgfgraphicnamed{scalars-s//RX21zero}
\InputIfFileExists{scalars-s//RX21zero.tikz}{}{\input{./figures/scalars-s//RX21zero.tikz}}
\endpgfgraphicnamed}^{(S3'L)} \otimes \intf{%
\beginpgfgraphicnamed{scalars-s//RX21zero}
\InputIfFileExists{scalars-s//RX21zero.tikz}{}{\input{./figures/scalars-s//RX21zero.tikz}}
\endpgfgraphicnamed}^{(S3'L)} \right] \circ \right. \\
  &\qquad \left. \left[ \intf{~\didf~}^{(S3'L)} \otimes \intf{~\dsigmaf~}^{(S3'L)} \otimes \intf{~\didf~}^{(S3'L)} \right] \circ \left[ \intf{%
\beginpgfgraphicnamed{scalars-s//RZ12zero}
\InputIfFileExists{scalars-s//RZ12zero.tikz}{}{\input{./figures/scalars-s//RZ12zero.tikz}}
\endpgfgraphicnamed}^{(S3'L)} \otimes \intf{%
\beginpgfgraphicnamed{scalars-s//RZ12zero}
\InputIfFileExists{scalars-s//RZ12zero.tikz}{}{\input{./figures/scalars-s//RZ12zero.tikz}}
\endpgfgraphicnamed}^{(S3'L)} \right] \right) \\
  &= 1 \otimes \left( \left[ \begin{pmatrix}1&0&0&0\\0&0&0&1\end{pmatrix} \otimes \begin{pmatrix}1&0&0&0\\0&0&0&1\end{pmatrix} \right] \circ \right. \\
  &\qquad \left. \left[ \begin{pmatrix}1&0\\0&1\end{pmatrix} \otimes \begin{pmatrix}1&0&0&0\\0&0&1&0\\0&1&0&0\\0&0&0&1\end{pmatrix} \otimes \begin{pmatrix}1&0\\0&1\end{pmatrix} \right] \circ \left[ \begin{pmatrix} 1&0\\ 0&0\\ 0&0\\ 0&1 \end{pmatrix} \otimes \begin{pmatrix} 1&0\\ 0&0\\ 0&0\\ 0&1 \end{pmatrix} \right] \right).
\end{align*}
To avoid having to work out matrix products with 16 rows or columns, consider the following argument. Recall that the tensor product of logical matrix is the same as the tensor product of standard matrices. Furthermore, it is straightforward to check that for binary-valued matrices with no more than a single 1 in each row and column, the logical matrix product is the same as the standard matrix product. Thus, the above matrix expression is equal to the \emph{standard interpretation} of a diagram consisting of green spiders, identities, and a swap. Using the standard interpretation means we can apply the rules of $\zxs$ to simplify the diagram before evaluating the interpretation:
\begin{align*}
 \intf{%
\beginpgfgraphicnamed{scalars//b2snew-l}
\InputIfFileExists{scalars//b2snew-l.tikz}{}{\input{./figures/scalars//b2snew-l.tikz}}
\endpgfgraphicnamed}^{(S3'L)} &= \intf{%
\beginpgfgraphicnamed{scalars//b2gn-l}
\InputIfFileExists{scalars//b2gn-l.tikz}{}{\input{./figures/scalars//b2gn-l.tikz}}
\endpgfgraphicnamed} = \intf{%
\beginpgfgraphicnamed{scalars//RZ22z}
\InputIfFileExists{scalars//RZ22z.tikz}{}{\input{./figures/scalars//RZ22z.tikz}}
\endpgfgraphicnamed} = \begin{pmatrix}1&0&0&0\\0&0&0&0\\0&0&0&0\\0&0&0&1\end{pmatrix}.
  \intertext{For the RHS we find}
  \intf{%
\beginpgfgraphicnamed{scalars//b2snew-r}
\InputIfFileExists{scalars//b2snew-r.tikz}{}{\input{./figures/scalars//b2snew-r.tikz}}
\endpgfgraphicnamed}^{(S3'L)} &= \left[ \intf{%
\beginpgfgraphicnamed{scalars-s//RZ10z}
\begin{tikzpicture}[yscale=-1]
	\begin{pgfonlayer}{nodelayer}
		\node [style=gn] (0) at (0.5, 0.1) {};
		\node [style=none] (1) at (0.5, -0.2) {};
	\end{pgfonlayer}
	\begin{pgfonlayer}{edgelayer}
		\draw (1.center) to (0.center);
	\end{pgfonlayer}

\end{tikzpicture}
}
\endpgfgraphicnamed}^{(S3'L)} \circ \intf{%
\beginpgfgraphicnamed{scalars-s//RX01z}
\begin{tikzpicture}
	\begin{pgfonlayer}{nodelayer}
		\node [style=rn] (0) at (0.5, 0.1) {};
		\node [style=none] (1) at (0.5, -0.2) {};
	\end{pgfonlayer}
	\begin{pgfonlayer}{edgelayer}
		\draw (1.center) to (0.center);
	\end{pgfonlayer}

\end{tikzpicture}
}
\endpgfgraphicnamed}^{(S3'L)} \right] \otimes \left[ \intf{%
\beginpgfgraphicnamed{scalars-s//RZ12zero}
\InputIfFileExists{scalars-s//RZ12zero.tikz}{}{\input{./figures/scalars-s//RZ12zero.tikz}}
\endpgfgraphicnamed}^{(S3'L)} \circ \intf{%
\beginpgfgraphicnamed{scalars-s//RX21zero}
\InputIfFileExists{scalars-s//RX21zero.tikz}{}{\input{./figures/scalars-s//RX21zero.tikz}}
\endpgfgraphicnamed}^{(S3'L)} \right] \\
 &= \left[ \begin{pmatrix}1&1\end{pmatrix} \circ \begin{pmatrix}1\\1\end{pmatrix} \right] \otimes \left[ \begin{pmatrix} 1&0\\ 0&0\\ 0&0\\ 0&1 \end{pmatrix} \circ \begin{pmatrix}1&0&0&0\\0&0&0&1\end{pmatrix} \right]
 = \begin{pmatrix}1&0&0&0\\0&0&0&0\\0&0&0&0\\0&0&0&1\end{pmatrix},
 \end{align*}
 hence (B2$'$) is sound.
 Now,
 \[
  \intf{~%
\beginpgfgraphicnamed{scalars//b1s-l}
\InputIfFileExists{scalars//b1s-l.tikz}{}{\input{./figures/scalars//b1s-l.tikz}}
\endpgfgraphicnamed~}^{(B1)} = \left[ \begin{pmatrix}1&1\end{pmatrix} \circ \begin{pmatrix}1\\1\end{pmatrix} \right] \otimes \left[ \begin{pmatrix} 1&0\\ 0&0\\ 0&0\\ 0&1 \end{pmatrix} \circ \begin{pmatrix}1\\1\end{pmatrix}\right] = 1\otimes  \begin{pmatrix}1\\0\\0\\1\end{pmatrix} = \begin{pmatrix}1\\0\\0\\1\end{pmatrix}
 \]
 since it is a scaling of a connected diagram, but
 \[
  \intf{~%
\beginpgfgraphicnamed{scalars//b1s-r}
\begin{tikzpicture}
	\begin{pgfonlayer}{nodelayer}
		\node [style=none] (1) at (2.25, -0.25) {};
		\node [style=rn] (3) at (2.25, 0.25) {};
		\node [style=rn] (6) at (2.75, 0.25) {};
		\node [style=none] (7) at (2.75, -0.25) {};
	\end{pgfonlayer}
	\begin{pgfonlayer}{edgelayer}
		\draw [style=none] (3) to (1.center);
		\draw [style=none] (6) to (7.center);
	\end{pgfonlayer}
\end{tikzpicture}
}
\endpgfgraphicnamed~}^{(B1)} = \begin{pmatrix}1\\1\end{pmatrix} \otimes \begin{pmatrix}1\\1\end{pmatrix} = \begin{pmatrix}1\\1\\1\\1\end{pmatrix},
 \]
 so (B1) is not sound.
 Thus, (B1) cannot be derived from the other rules.
\end{proof}

The (EU$'$) rule, which can be interpreted as the Euler decomposition of H, was not present in the seminal paper by Coecke and Duncan \cite{coecke_interacting_2011}. It has been proved later on that the Euler decomposition of H cannot be derived from the original rules of the \ZX-calculus, this rule was then added to the theory \cite{DP09}. The (EU$'$) rule is also necessary in the simplified stabilizer \ZX-calculus:

\begin{lem}
 The  (EU$'$) rule is necessary:  $\zxs \setminus \textup{(EU$'$)} \not \vdash  \textup{(EU$'$)}$. 
\end{lem}

\begin{proof}
The original proof in \cite{DP09,duncan_pivoting_2014} that the Euler decomposition is necessary does not directly apply here since the set of rules is different, and actually the Euler rule is also different. Our proof is however similar,  consisting in `doubling' the diagram in such a way that each dot is encoded using  two dots -- one of each colour -- 
and each H is encoded as a swap.

For any diagram $D:n\to m$, let $\intf{D}^{(EU')}\in \mathbb C^{2^{2n}\times 2^{2m}}$ be inductively defined as $\intf{D_1\otimes D_2}^{(EU')} = \intf{D_1}^{(EU')}{\otimes} \intf{D_2}^{(EU')}$, $\intf{D_2\circ D_1}^{(EU')} = \intf{D_2}^{(EU')}{\circ}\intf{D_1}^{(EU')}$,
\begin{align*}
 \intf{~\dempty~}^{(EU')} &=1 &
 \intf{~\didf~}^{(EU')}  &=\intf{~\did~~\did~} &
 \intf{~\dHf~}^{(EU')} &= \intf{~\dsigmaf~} \\
 \intf{~\dsigmaf~}^{(EU')} &= \intf{\dsigmaf\!\!\!\!\!\!\!\dsigmaf} &
 \intf{\drcupf}^{(EU')} &=  \intf{\drcup\!\!\!\!\!\!\!\!\!\!\!\!\!\!\!\!\!\drcup} &
 \intf{\drcapf}^{(EU')} &=\intf{\drcapf\!\!\!\!\!\!\!\!\!\!\!\!\!\!\!\!\!\drcapf} \\
 \intf{~%
\beginpgfgraphicnamed{scalars-s/spiderg-a}
\InputIfFileExists{scalars-s/spiderg-a.tikz}{}{\input{./figures/scalars-s/spiderg-a.tikz}}
\endpgfgraphicnamed~}^{(EU')} &= \intf{%
\beginpgfgraphicnamed{scalars-s/spiderinterpretationsc-2}
\InputIfFileExists{scalars-s/spiderinterpretationsc-2.tikz}{}{\input{./figures/scalars-s/spiderinterpretationsc-2.tikz}}
\endpgfgraphicnamed} & & &
 \intf{~%
\beginpgfgraphicnamed{scalars-s/spiderr-a}
\InputIfFileExists{scalars-s/spiderr-a.tikz}{}{\input{./figures/scalars-s/spiderr-a.tikz}}
\endpgfgraphicnamed~}^{(EU')} &= \intf{%
\beginpgfgraphicnamed{scalars-s/rspiderinterpretationsc-2}
\InputIfFileExists{scalars-s/rspiderinterpretationsc-2.tikz}{}{\input{./figures/scalars-s/rspiderinterpretationsc-2.tikz}}
\endpgfgraphicnamed}
\end{align*}
All the H-free rules are sound with respect to $\intf{.}^{(EU')}$, as this interpretation essentially corresponds to doubling H-free diagrams.
The (H) rule is also satisfied as exchanging the two copies of the encoding dots  corresponds to exchanging the colour of the encoded dot. The connectivity meta rule is sound as the interpretation maps \ZX-calculus diagrams to \ZX-calculus diagrams, which again satisfy the connectivity meta rule.
 
 On the other hand, the (EU$'$) is not sound with respect to the $\intf{.}^{(EU')}$ interpretation:
 \[
  \intf{%
\beginpgfgraphicnamed{scalars-s//eulernewanglepart}
\InputIfFileExists{scalars-s//eulernewanglepart.tikz}{}{\input{./figures/scalars-s//eulernewanglepart.tikz}}
\endpgfgraphicnamed}^{(EU')} = \intf{%
\beginpgfgraphicnamed{scalars-s//eulernewanglepartdoubled}
\InputIfFileExists{scalars-s//eulernewanglepartdoubled.tikz}{}{\input{./figures/scalars-s//eulernewanglepartdoubled.tikz}}
\endpgfgraphicnamed} = \intf{%
\beginpgfgraphicnamed{scalars-s//eulernewanglepartdoubled2}
\InputIfFileExists{scalars-s//eulernewanglepartdoubled2.tikz}{}{\input{./figures/scalars-s//eulernewanglepartdoubled2.tikz}}
\endpgfgraphicnamed} = \intf{~\dHf~~\dHf~} \neq \intf{~\dsigmaf~} = \intf{~\dHf~}^{(EU')}
 \]
 As a consequence (EU$'$) cannot be derived from the other rules of the language.
\end{proof}

(H) is another fundamental rule of the \ZX-calculus which states that H can be used to change the colour of dot. This rule is also necessary: 

\begin{lem}
 The  (H) rule is necessary:  $\zxs \setminus \textup{(H)} \not \vdash  \textup{(H)}$. 
\end{lem}

\begin{proof}
All rules but (H) are sound for the following interpretation: for any diagram $D$, let $\denote D^{(H)}\in \mathbb R$ be inductively defined as $\denote {D_1\otimes D_2}^{(H)} =\denote {D_1\circ  D_2}^{(H)} = \denote {D_1}^{(H)}\denote {D_2}^{(H)}$,
\[
 \intf{~%
\beginpgfgraphicnamed{scalars-s/spiderr-a}
\InputIfFileExists{scalars-s/spiderr-a.tikz}{}{\input{./figures/scalars-s/spiderr-a.tikz}}
\endpgfgraphicnamed~}^{(H)}=1-\alpha,
\]
and $\denote{.}^{(H)}=1$ for all the other generators. All the rules which do not involve a red dot with a non-zero angle are trivially sound with respect to this interpretation. The connectivity meta rule is sound since the interpretation depends only on the phase labels, which must be consistent across diagrams that are related by graph isomorphism.
 It only remains (H) which is not sound e.g. when $\alpha=\frac \pi 2$.
 As a consequence (H) cannot be derived from the other rules.
\end{proof}

The (IV$'$) rule, like all the rules dealing with scalars (i.e.\ subdiagrams with no inputs and no outputs), has been introduced more recently \cite{backens_making_2015}.

\begin{lem}
 The  (IV$'$) rule is necessary:  $\zxs \setminus \textup{(IV$'$)} \not \vdash  \textup{(IV$'$)}$. 
\end{lem}

\begin{proof}
Following \cite{BPW16} (Section 3.3), one can notice that (IV$'$) is the only rule which equates an empty diagram and a non empty diagram and  thus (IV$'$) cannot be derived from the other rules. More formally, for any diagram $D$, define the invariant $\intf{D}^{(IV')}\in \{0,1\}$, as follows: let $G_D$ be the labelled graph corresponding to $D$ according to Definition~\ref{dfn:labelled_graph} and let
\[
 \intf{D}^{(IV')} = \begin{cases} 1 &\text{if $G_D$ is the empty graph} \\ 0 & \text{otherwise}. \end{cases}
\]
This is an invariant of labelled graph isomorphisms, so if two diagrams have the same connectivity, they have the same value of $\intf{.}^{(IV')}$.
All the explicit rewrite rules except (IV$'$) preserve this invariant since they map non-empty diagrams to non-empty diagrams.
Therefore, (IV$'$) cannot be derived from the other rules.
\end{proof}

The absorbing behaviour of the scalar zero, represented by the diagram \gpi,  is captured by the (ZO$'$) rule.

\begin{lem}\label{lemma:B2S3R}
 The  (ZO$'$) rule is necessary:  $\zxs \setminus \textup{(ZO$'$)} \not \vdash  \textup{(ZO$'$)}$. 
\end{lem}

\begin{proof}
Let $\denote D^{(ZO')}\in \{0,1\}$ be the parity of the number of $H$ plus the number of odd-degree red dots.
More formally, if $D$ is a \ZX-diagram with corresponding labelled graph $G_D = (V,E,\ell)$, then let
\[
 \intf{D}^{(ZO')} = \left( \big|\{v\in V|\ell(v)=H\}\big| + \sum_{\{v\in V|\ell(v)=R_X(\alpha)\}} \deg(v) \right) \bmod 2,
\]
where the first sum is over all vertices labelled $H$ and the second sum is over all vertices whose label is of the form $R_X(\alpha)$ for some angle $\alpha$.
This property is invariant under labelled graph isomorphisms, so if two diagrams have the same connectivity, they have the same value of $\intf{.}^{(ZO')}$.

Furthermore, all the explicit rewrite rules except (ZO$'$) are sound under this interpretation: For (S1) and (S3$'$L) this is trivial, since they involve neither red dots nor $H$.
For (S3$'$R), (B1), (B2$'$) and (IV$'$), the total degree of red dots has the same parity on both sides.
For (EU$'$) and (H), the total degree of red dots plus the number of $H$ boxes has the same parity on both sides, e.g.\ for (EU$'$):
\[
 \intf{%
\beginpgfgraphicnamed{scalars-s//eulernewanglepart}
\InputIfFileExists{scalars-s//eulernewanglepart.tikz}{}{\input{./figures/scalars-s//eulernewanglepart.tikz}}
\endpgfgraphicnamed}^{(ZO')} = 3\bmod 2 = 1 = \intf{~\dHf~}^{(ZO')}.
\]
Yet for (ZO$'$), we have
\[
 \intf{%
\beginpgfgraphicnamed{scalars//zo1-prime-l}
\begin{tikzpicture}
	\begin{pgfonlayer}{nodelayer}
		\node [style=gn] (0) at (-1, 0) {$~\pi~$};
		\node [style=none] (1) at (-0.5, -0.25) {};
		\node [style=gn] (2) at (-0.5, 0.25) {};
	\end{pgfonlayer}
	\begin{pgfonlayer}{edgelayer}
		\draw (1.center) to (2);
	\end{pgfonlayer}
\end{tikzpicture}}
\endpgfgraphicnamed}^{(ZO')} = 0 \neq 1 = \intf{%
\beginpgfgraphicnamed{scalars//zo1-prime-r}
\begin{tikzpicture}
	\begin{pgfonlayer}{nodelayer}
		\node [style=rn] (0) at (1, 0.25) {};
		\node [style=none] (1) at (1, -0.25) {};
		\node [style=gn] (2) at (0.5, 0) {$~\pi~$};
	\end{pgfonlayer}
	\begin{pgfonlayer}{edgelayer}
		\draw (0) to (1.center);
	\end{pgfonlayer}
\end{tikzpicture}}
\endpgfgraphicnamed}^{(ZO')}.
\]
As a consequence (ZO$'$) cannot be derived from the other rules. 
\end{proof}

Finally, among the 9 rules of the simplified \ZX-calculus, the necessity of two rules -- (B2$'$) and (S3$'$R) -- remains unknown. We can however prove that at least one of the two is necessary: 

\begin{lem}
Either  (B2$'$) or (S3$'$R) is necessary:
\[
 \zxs \setminus \{\textup{(B2$'$), (S3$'$R)\}} \not \vdash  \textup{(B2$'$)} \qquad\text{and}\qquad
 \zxs \setminus \{\textup{(B2$'$), (S3$'$R)\}} \not \vdash  \textup{(S3$'$R)}.
\]
\end{lem}

\begin{proof}
All rules except (B2$'$) and (S3$'$R) are sound with respect to  the following interpretation: for any diagram $D:n\to m$, let $\denote D^{\natural}\in \mathbb C^{2^n\times 2^m}$ be inductively defined as $\denote {D_1\otimes D_2}^\natural =   \denote {D_1}^\natural\otimes \denote {D_2}^\natural$,  $\denote {D_1\circ D_2}^\natural = \denote {D_1}^\natural\circ \denote {D_2}^\natural$,
\[
 \intf{%
\beginpgfgraphicnamed{scalars//spideralpha-s}
\InputIfFileExists{scalars//spideralpha-s.tikz}{}{\input{./figures/scalars//spideralpha-s.tikz}}
\endpgfgraphicnamed}^\natural = e^{\frac 43i\alpha}\intf{%
\beginpgfgraphicnamed{scalars//spideralpha-s}
\InputIfFileExists{scalars//spideralpha-s.tikz}{}{\input{./figures/scalars//spideralpha-s.tikz}}
\endpgfgraphicnamed}, \qquad
 \intf{%
\beginpgfgraphicnamed{scalars//spiderredalpha-s}
\InputIfFileExists{scalars//spiderredalpha-s.tikz}{}{\input{./figures/scalars//spiderredalpha-s.tikz}}
\endpgfgraphicnamed}^\natural = e^{\frac 43i\alpha+(n+m)i\frac \pi 3}\intf{%
\beginpgfgraphicnamed{scalars//spiderredalpha-s}
\InputIfFileExists{scalars//spiderredalpha-s.tikz}{}{\input{./figures/scalars//spiderredalpha-s.tikz}}
\endpgfgraphicnamed},
\]
$\intf{~\dHf~}^\natural  =  e^{-i\frac\pi 3}  \intf{~\dHf~}$, and $\denote{.}^\natural=\denote{.}$ for all the other generators.

The connectivity meta rule is sound under this interpretation as it only multiplies the original interpretation by some scalar that depends on the phase labels.
Most of the explicit rewrite rules are also sound.
For (S1), (S3$'$L) and (B1), we have:
\[
 \intf{%
\beginpgfgraphicnamed{scalars//spider-bis-l}
\InputIfFileExists{scalars//spider-bis-l.tikz}{}{\input{./figures/scalars//spider-bis-l.tikz}}
\endpgfgraphicnamed}^\natural
 = e^{\frac 43i\alpha} e^{\frac 43i\beta} \intf{%
\beginpgfgraphicnamed{scalars//spider-bis-l}
\InputIfFileExists{scalars//spider-bis-l.tikz}{}{\input{./figures/scalars//spider-bis-l.tikz}}
\endpgfgraphicnamed}
 = e^{\frac 43i(\alpha+\beta)} \intf{%
\beginpgfgraphicnamed{scalars//spider-bis-r}
\InputIfFileExists{scalars//spider-bis-r.tikz}{}{\input{./figures/scalars//spider-bis-r.tikz}}
\endpgfgraphicnamed}
 = \intf{%
\beginpgfgraphicnamed{scalars//spider-bis-r}
\InputIfFileExists{scalars//spider-bis-r.tikz}{}{\input{./figures/scalars//spider-bis-r.tikz}}
\endpgfgraphicnamed}^\natural
\]
\[
 \intf{\drcapf}^\natural = \intf{\drcapf} = \intf{%
\beginpgfgraphicnamed{scalars//green_cap}
\InputIfFileExists{scalars//green_cap.tikz}{}{\input{./figures/scalars//green_cap.tikz}}
\endpgfgraphicnamed}  = \intf{%
\beginpgfgraphicnamed{scalars//green_cap}
\InputIfFileExists{scalars//green_cap.tikz}{}{\input{./figures/scalars//green_cap.tikz}}
\endpgfgraphicnamed}^\natural
\]
\[
 \intf{%
\beginpgfgraphicnamed{scalars//b1s-l}
\InputIfFileExists{scalars//b1s-l.tikz}{}{\input{./figures/scalars//b1s-l.tikz}}
\endpgfgraphicnamed}^\natural = e^{2i\frac{\pi}{3}} \intf{%
\beginpgfgraphicnamed{scalars//b1s-l}
\InputIfFileExists{scalars//b1s-l.tikz}{}{\input{./figures/scalars//b1s-l.tikz}}
\endpgfgraphicnamed} = e^{2i\frac{\pi}{3}} \intf{%
\beginpgfgraphicnamed{scalars//b1s-r}
\begin{tikzpicture}
	\begin{pgfonlayer}{nodelayer}
		\node [style=none] (1) at (2.25, -0.25) {};
		\node [style=rn] (3) at (2.25, 0.25) {};
		\node [style=rn] (6) at (2.75, 0.25) {};
		\node [style=none] (7) at (2.75, -0.25) {};
	\end{pgfonlayer}
	\begin{pgfonlayer}{edgelayer}
		\draw [style=none] (3) to (1.center);
		\draw [style=none] (6) to (7.center);
	\end{pgfonlayer}
\end{tikzpicture}
}
\endpgfgraphicnamed} = \intf{%
\beginpgfgraphicnamed{scalars//b1s-r}
\begin{tikzpicture}
	\begin{pgfonlayer}{nodelayer}
		\node [style=none] (1) at (2.25, -0.25) {};
		\node [style=rn] (3) at (2.25, 0.25) {};
		\node [style=rn] (6) at (2.75, 0.25) {};
		\node [style=none] (7) at (2.75, -0.25) {};
	\end{pgfonlayer}
	\begin{pgfonlayer}{edgelayer}
		\draw [style=none] (3) to (1.center);
		\draw [style=none] (6) to (7.center);
	\end{pgfonlayer}
\end{tikzpicture}
}
\endpgfgraphicnamed}^\natural
\]
The case of (EU$'$) also works out since $-i\frac\pi 3$ differs from $\frac 43 i\frac{\pi}2 + \frac 43 i\frac{\pi}2 + \frac 43 i\frac{-\pi}2 + 3i\frac{\pi}{3} = \frac 53 i\pi$ by exactly $2\pi$:
\[
 \intf{~\dHf~}^\natural = e^{-i\frac\pi 3}  \intf{~\dHf~} = e^{-i\frac\pi 3}  \intf{%
\beginpgfgraphicnamed{scalars-s//eulernewanglepart}
\InputIfFileExists{scalars-s//eulernewanglepart.tikz}{}{\input{./figures/scalars-s//eulernewanglepart.tikz}}
\endpgfgraphicnamed} = e^{\frac 43 i\frac{\pi}2} e^{\frac 43 i\frac{\pi}2} e^{\frac 43 i\frac{-\pi}2} e^{3i\frac \pi 3} \intf{%
\beginpgfgraphicnamed{scalars-s//eulernewanglepart}
\InputIfFileExists{scalars-s//eulernewanglepart.tikz}{}{\input{./figures/scalars-s//eulernewanglepart.tikz}}
\endpgfgraphicnamed}  = \intf{%
\beginpgfgraphicnamed{scalars-s//eulernewanglepart}
\InputIfFileExists{scalars-s//eulernewanglepart.tikz}{}{\input{./figures/scalars-s//eulernewanglepart.tikz}}
\endpgfgraphicnamed}^\natural
\]
For (H), the scalars resulting from the $H$ boxes exactly balance out the scalar factor resulting from the arity of the red dot.
The rule (IV$'$) becomes:
\[
 \intf{%
\beginpgfgraphicnamed{scalars//dotinverse-l}
\InputIfFileExists{scalars//dotinverse-l.tikz}{}{\input{./figures/scalars//dotinverse-l.tikz}}
\endpgfgraphicnamed}^\natural = e^{3i\frac \pi 3}e^{3i\frac \pi 3} \intf{%
\beginpgfgraphicnamed{scalars//dotinverse-l}
\InputIfFileExists{scalars//dotinverse-l.tikz}{}{\input{./figures/scalars//dotinverse-l.tikz}}
\endpgfgraphicnamed} = \intf{%
\beginpgfgraphicnamed{scalars//dotinverse-l}
\InputIfFileExists{scalars//dotinverse-l.tikz}{}{\input{./figures/scalars//dotinverse-l.tikz}}
\endpgfgraphicnamed} = \intf{%
\beginpgfgraphicnamed{scalars-s//emptysquare-small}
\InputIfFileExists{scalars-s//emptysquare-small.tikz}{}{\input{./figures/scalars-s//emptysquare-small.tikz}}
\endpgfgraphicnamed} = \intf{%
\beginpgfgraphicnamed{scalars-s//emptysquare-small}
\InputIfFileExists{scalars-s//emptysquare-small.tikz}{}{\input{./figures/scalars-s//emptysquare-small.tikz}}
\endpgfgraphicnamed}^\natural
\]
In (ZO$'$), both sides of the equality are interpreted as a zero map, so scalar multiplication has no effect and the rule remains sound.

Yet (S3$'$R) and (B2$'$) are not sound under $\intf{.}^\natural$:
\[
 \intf{%
\beginpgfgraphicnamed{scalars//green_cap}
\InputIfFileExists{scalars//green_cap.tikz}{}{\input{./figures/scalars//green_cap.tikz}}
\endpgfgraphicnamed}^\natural = \intf{%
\beginpgfgraphicnamed{scalars//green_cap}
\InputIfFileExists{scalars//green_cap.tikz}{}{\input{./figures/scalars//green_cap.tikz}}
\endpgfgraphicnamed} \neq e^{2i\frac\pi 3} \intf{%
\beginpgfgraphicnamed{scalars//red_cap}
\InputIfFileExists{scalars//red_cap.tikz}{}{\input{./figures/scalars//red_cap.tikz}}
\endpgfgraphicnamed} = \intf{%
\beginpgfgraphicnamed{scalars//red_cap}
\InputIfFileExists{scalars//red_cap.tikz}{}{\input{./figures/scalars//red_cap.tikz}}
\endpgfgraphicnamed}^\natural
\]
\[
 \intf{%
\beginpgfgraphicnamed{scalars//b2s-l}
\InputIfFileExists{scalars//b2s-l.tikz}{}{\input{./figures/scalars//b2s-l.tikz}}
\endpgfgraphicnamed}^\natural = e^{i\frac\pi 3} e^{3i\frac \pi 3}e^{3i\frac \pi 3}  \intf{%
\beginpgfgraphicnamed{scalars//b2s-l}
\InputIfFileExists{scalars//b2s-l.tikz}{}{\input{./figures/scalars//b2s-l.tikz}}
\endpgfgraphicnamed} = e^{i\frac\pi 3} \intf{%
\beginpgfgraphicnamed{scalars//b2s-l}
\InputIfFileExists{scalars//b2s-l.tikz}{}{\input{./figures/scalars//b2s-l.tikz}}
\endpgfgraphicnamed} \neq e^{3i\frac\pi 3} \intf{%
\beginpgfgraphicnamed{scalars//b2s-r}
\InputIfFileExists{scalars//b2s-r.tikz}{}{\input{./figures/scalars//b2s-r.tikz}}
\endpgfgraphicnamed} = \intf{%
\beginpgfgraphicnamed{scalars//b2s-r}
\InputIfFileExists{scalars//b2s-r.tikz}{}{\input{./figures/scalars//b2s-r.tikz}}
\endpgfgraphicnamed}^\natural
\]
Thus, (S3$'$R) and (B2$'$) cannot be derived from $\zxs\setminus\{$(S3$'$R), (B2$'$)$\}$.
\end{proof}

The two parts of (S3$'$) are very similar, so it is understandable that it would be difficult to determine whether they are independent of each other. 
It is more vexing not to be able to prove whether the bialgebra rule (B2$'$) is necessary. Indeed the bialgebra rule (B2$'$) plays a central role in the language: it is the cornerstone of the axiomatisation of complementary bases.  Thus, it would be unexpected for the bialgebra rule to be derivable from the other rules. In fact, the rewrite rules can be modified to make (B2$'$) the only rule that is not sound under $\intf{\cdot}^{\flat}$, as detailed below in Remark~\ref{rem:bialgebra_necessity}. Yet this comes at the cost of introducing additional scalars in several rules, which adds gratuitous complexity and also invalidates the necessity proof for (S3$'$L). 

While the bialgebra rule (B2$'$) is at the heart of the characterisation of complementary bases, the interpretation of the (S3$'$R) rule is that the two bases -- one characterised by the green dots, the other by the red dots -- are inducing the same compact structure. Indeed, each colour is inducing a compact structure, i.e.\ a pair of a `cup' and a `cap' that satisfy a `snake equation' like in Figure \ref{fig:compact_structure}. There is no a priori reason that those two compact structures should coincide.
Thus, deciding whether (S3$'$R) is necessary is related to the question of deciding whether the other rules of the language force the compact structures induced by the green and the red dots, respectively, to coincide.

\begin{rem}\label{rem:bialgebra_necessity}
The bialgebra rule (B2$'$) can be made necessary while retaining soundness and completeness by modifying two of the other rewrite rules as follows.

Replace (S3$'$) by (S3) and the following rule:
\begin{equation}\label{eq:S3tilde}
\beginpgfgraphicnamed{scalars//S3_tilde}
\InputIfFileExists{scalars//S3_tilde.tikz}{}{\input{./figures/scalars//S3_tilde.tikz}}
\endpgfgraphicnamed
\end{equation}
Additionally, replace (IV$'$) by:
\begin{equation}
\beginpgfgraphicnamed{scalars//IV_tilde2}
\InputIfFileExists{scalars//IV_tilde2.tikz}{}{\input{./figures/scalars//IV_tilde2.tikz}}
\endpgfgraphicnamed
\end{equation}
where the right-hand side denotes an empty diagram.

In the resulting rule set, (B2$'$) is the only rule that is not sound under the interpretation functor $\intf{\cdot}^{\flat}$ which acts like the usual interpretation functor on green dots, wires, and the empty diagram, but adds complex phases to red dots (depending on their degree) and to Hadamard nodes:
 \[
  \intf{%
\beginpgfgraphicnamed{scalars//spiderredalpha}
\InputIfFileExists{scalars//spiderredalpha.tikz}{}{\input{./figures/scalars//spiderredalpha.tikz}}
\endpgfgraphicnamed}^{\flat} = i^{m+n} \intf{%
\beginpgfgraphicnamed{scalars//spiderredalpha}
\InputIfFileExists{scalars//spiderredalpha.tikz}{}{\input{./figures/scalars//spiderredalpha.tikz}}
\endpgfgraphicnamed} \qquad\text{and}\qquad
  \intf{%
\beginpgfgraphicnamed{scalars//Hadamard}
\begin{tikzpicture}
	\begin{pgfonlayer}{nodelayer}
		\node [style={H box}] (0) at (0, 0) {};
		\node [style=none] (1) at (0, 0.5) {};
		\node [style=none] (2) at (0, -0.5) {};
	\end{pgfonlayer}
	\begin{pgfonlayer}{edgelayer}
		\draw (2.center) to (1.center);
	\end{pgfonlayer}
\end{tikzpicture}
}
\endpgfgraphicnamed}^{\flat} = -i \intf{%
\beginpgfgraphicnamed{scalars//Hadamard}
\begin{tikzpicture}
	\begin{pgfonlayer}{nodelayer}
		\node [style={H box}] (0) at (0, 0) {};
		\node [style=none] (1) at (0, 0.5) {};
		\node [style=none] (2) at (0, -0.5) {};
	\end{pgfonlayer}
	\begin{pgfonlayer}{edgelayer}
		\draw (2.center) to (1.center);
	\end{pgfonlayer}
\end{tikzpicture}
}
\endpgfgraphicnamed}.
 \]

The rule replacing (IV$'$) is necessary by the same argument as (IV$'$) itself. Additionally, by the argument in Lemma \ref{lem:disconnecting_functor}, at least one of (S3) and \eqref{eq:S3tilde} is necessary, but it is unclear whether they both are.
\end{rem}

\section{Simplifying the ambient category}
\label{s:simplified_category}

As shown in \cite{BPW16}, $\zxs$ is complete without the need of assuming that the colour-swap and upside-down rules are also satisfied.  However, the meta-rule `only the connectivity matters' was supposed to hold.
We now consider how to replace this powerful meta rule with weaker assumptions based on the graphical axioms for specific categories. Indeed the `only connectivity matters rule' is actually a combination of axioms making the ambient category compact closed (which implies that isomorphic diagrams are equal), together with some extra properties of the generators (e.g. commutativity or partial transpose). We show in the following that these extra properties of the generators, can essentially be derived from the properties of the ambient category (section \ref{s:compact_closed}), even when the ambient category is braided (section \ref{s:braided_category}).

\subsection{Compact closed category / Isomorphism}
\label{s:compact_closed}

The standard route \cite{coecke_interacting_2011} for axiomatising graphical properties like `only the connectivity matters' in a categorical framework is based on compact closed categories \cite{selinger_dagger_2007,selinger_survey_2010}.  Assuming that we work with a compact closed category  means assuming that the equations in Figure \ref{fig:compact_structure} are satisfied. It additionally implies that arbitrary maps can slide freely along either wire in a crossing. Graphically, this means that any two isomorphic diagrams are equal. It is straightforward to check that all of the above rules are sound for the \ZX-calculus with the `only connectivity matters rule'.

\begin{figure}
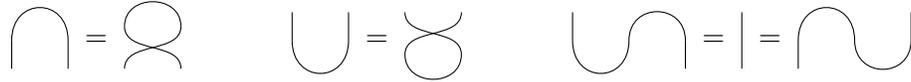

 \centering
\beginpgfgraphicnamed{scalars//compactstructure_cap}
\InputIfFileExists{scalars//compactstructure_cap.tikz}{}{\input{./figures/scalars//compactstructure_cap.tikz}}
\endpgfgraphicnamed \qquad\quad %
\beginpgfgraphicnamed{scalars//compactstructure_cup}
\InputIfFileExists{scalars//compactstructure_cup.tikz}{}{\input{./figures/scalars//compactstructure_cup.tikz}}
\endpgfgraphicnamed \qquad\quad %
\beginpgfgraphicnamed{scalars//compactstructure_snake}
\InputIfFileExists{scalars//compactstructure_snake.tikz}{}{\input{./figures/scalars//compactstructure_snake.tikz}}
\endpgfgraphicnamed
 \caption{Some of the equations satisfied by the structural maps in a compact closed category: the caps and cups are symmetrical and they satisfy the \emph{snake equations}.}\label{fig:compact_structure}
\end{figure}

At first sight, it seems like the compact closed structure is significantly less powerful than the `only connectivity matters rule': in particular, working in a compact closed category does not directly imply any symmetry properties for the nodes, like the ability to swap legs or bend inputs into outputs:
\begin{equation}\label{eq:node_symmetries}
\beginpgfgraphicnamed{scalars-s//commute1}
\InputIfFileExists{scalars-s//commute1.tikz}{}{\input{./figures/scalars-s//commute1.tikz}}
\endpgfgraphicnamed =~%
\beginpgfgraphicnamed{scalars-s//commute2}
\InputIfFileExists{scalars-s//commute2.tikz}{}{\input{./figures/scalars-s//commute2.tikz}}
\endpgfgraphicnamed~\qquad\qquad %
\beginpgfgraphicnamed{scalars-s//bendingnew}
\InputIfFileExists{scalars-s//bendingnew.tikz}{}{\input{./figures/scalars-s//bendingnew.tikz}}
\endpgfgraphicnamed ~=%
\beginpgfgraphicnamed{scalars-s//nonbending}
\InputIfFileExists{scalars-s//nonbending.tikz}{}{\input{./figures/scalars-s//nonbending.tikz}}
\endpgfgraphicnamed
\end{equation}

Nevertheless, it is possible to derive all of these properties using just one more rewrite rule in addition to the ones given in Figure \ref{figure1}, namely:
\begin{equation}\tag{S2$'$}\label{eq:s2prime}
\beginpgfgraphicnamed{scalars//s2red}
\InputIfFileExists{scalars//s2red.tikz}{}{\input{./figures/scalars//s2red.tikz}}
\endpgfgraphicnamed
\end{equation}

Let $\zxc$ be the calculus obtained by considering the rules of figure \ref{figure1} together with the (S2$'$) rule and the axioms of the ambient compact closed category.

Notice that in $\zxc$, the generators of the language are not supposed to be commutative, as a consequence the ellipsis notations ($\ldots$), like in (S1) and (H) do not involve any crossing of wires. 

We first prove three useful properties derivable in $\zxc$:

\begin{lem}\label{ruels2}
  $ \zxc \vdash~ %
\beginpgfgraphicnamed{scalars//GreenID-2}
\begin{tikzpicture}[scale=0.6]
	\begin{pgfonlayer}{nodelayer}
		\node [style=gn] (0) at (0, 0) {};
		\node [style=none] (1) at (0, 0.5) {};
		\node [style=none] (2) at (0, -0.5) {};
	\end{pgfonlayer}
	\begin{pgfonlayer}{edgelayer}
		\draw (0) to (2.center);
		\draw (0) to (1.center) ;
	\end{pgfonlayer}
\end{tikzpicture}
}
\endpgfgraphicnamed~=~ %
\beginpgfgraphicnamed{scalars//Id-2}
\begin{tikzpicture}[scale=0.6]
	\begin{pgfonlayer}{nodelayer}
		\node [style=none] (1) at (0, 0.5) {};
		\node [style=none] (2) at (0, -0.5) {};
	\end{pgfonlayer}
	\begin{pgfonlayer}{edgelayer}
		\draw (2.center) to (1.center);
	\end{pgfonlayer}
\end{tikzpicture}
}
\endpgfgraphicnamed$~, ~~~ 
  $ \zxc \vdash~ %
\beginpgfgraphicnamed{scalars//Gcup}
\begin{tikzpicture}[scale=0.6]
	\begin{pgfonlayer}{nodelayer}
		\node [style=none] (6) at (0.5, 0.25) {};
		\node [style=none] (7) at (-0.5, 0.25) {};
		\node [style=gn] (10) at (0, -0.25) {};
	\end{pgfonlayer}
	\begin{pgfonlayer}{edgelayer}
		\draw [in=150, out=-90] (7.center) to (10);
		\draw [in=-90, out=30] (10) to (6.center);
	\end{pgfonlayer}
\end{tikzpicture}
}
\endpgfgraphicnamed~=~ %
\beginpgfgraphicnamed{scalars//cup}
\begin{tikzpicture}[scale=0.6]
	\begin{pgfonlayer}{nodelayer}
		\node [style=none] (6) at (0.5, 0.25) {};
		\node [style=none] (7) at (-0.5, 0.25) {};
	\end{pgfonlayer}
	\begin{pgfonlayer}{edgelayer}
		\draw [in=-90, out=-90, looseness=2.00] (7.center) to (6.center);
	\end{pgfonlayer}
\end{tikzpicture}
}
\endpgfgraphicnamed$~, ~and~~ 
$\zxc \vdash~ %
\beginpgfgraphicnamed{scalars//H2-small}
\begin{tikzpicture}[scale=0.6]
	\begin{pgfonlayer}{nodelayer}
		\node [style=none] (1) at (0, 0.5) {};
		\node [style=none] (2) at (0, -0.5) {};
		\node [style=H box] (3) at (0, -0.25) {};
		\node [style=H box] (4) at (0, 0.25) {};
	\end{pgfonlayer}
	\begin{pgfonlayer}{edgelayer}
		\draw (2.center) to (1.center);
	\end{pgfonlayer}
\end{tikzpicture}
}
\endpgfgraphicnamed~=~ %
\beginpgfgraphicnamed{scalars//Id-2}
\begin{tikzpicture}[scale=0.6]
	\begin{pgfonlayer}{nodelayer}
		\node [style=none] (1) at (0, 0.5) {};
		\node [style=none] (2) at (0, -0.5) {};
	\end{pgfonlayer}
	\begin{pgfonlayer}{edgelayer}
		\draw (2.center) to (1.center);
	\end{pgfonlayer}
\end{tikzpicture}
}
\endpgfgraphicnamed$~.
\end{lem}

\begin{proof}
We have 
\begin{equation}\label{S0}
\beginpgfgraphicnamed{scalars//ruels2prf1}
\begin{tikzpicture}[scale=0.6]
	\begin{pgfonlayer}{nodelayer}
		\node [style=gn] (0) at (0.5, 0) {};
		\node [style=none] (1) at (0.5, 0.5) {};
		\node [style=none] (2) at (0.5, -0.5) {};
	\end{pgfonlayer}
	\begin{pgfonlayer}{edgelayer}
		\draw (0) to (2.center);
		\draw (0) to (1.center);
	\end{pgfonlayer}
\end{tikzpicture}}
\endpgfgraphicnamed\quad \stackrel{\textup{\tiny CCC}}{=}\quad%
\beginpgfgraphicnamed{scalars//ruels2prf2}
\InputIfFileExists{scalars//ruels2prf2.tikz}{}{\input{./figures/scalars//ruels2prf2.tikz}}
\endpgfgraphicnamed\quad \stackrel{\textup{\tiny (S3$'$)}}{=}\quad%
\beginpgfgraphicnamed{scalars//ruels2prf3}
\InputIfFileExists{scalars//ruels2prf3.tikz}{}{\input{./figures/scalars//ruels2prf3.tikz}}
\endpgfgraphicnamed\quad \stackrel{\textup{\tiny (S1)}}{=}\quad%
\beginpgfgraphicnamed{scalars//ruels2prf4}
\InputIfFileExists{scalars//ruels2prf4.tikz}{}{\input{./figures/scalars//ruels2prf4.tikz}}
\endpgfgraphicnamed\quad \stackrel{\textup{\tiny (S3$'$)}}{=}\quad%
\beginpgfgraphicnamed{scalars//ruels2prf5}
\InputIfFileExists{scalars//ruels2prf5.tikz}{}{\input{./figures/scalars//ruels2prf5.tikz}}
\endpgfgraphicnamed\quad \stackrel{\textup{\tiny CCC}}{=}\quad%
\beginpgfgraphicnamed{scalars//ruels2prf6}
\begin{tikzpicture}[scale=0.6]
	\begin{pgfonlayer}{nodelayer}
		\node [style=none] (0) at (0, -0.5) {};
		\node [style=none] (1) at (0, 0.5) {};
	\end{pgfonlayer}
	\begin{pgfonlayer}{edgelayer}
		\draw (1.center) to (0.center);
	\end{pgfonlayer}
\end{tikzpicture}
}
\endpgfgraphicnamed
 \end{equation}
 Moreover, 
  \begin{equation*} 
\beginpgfgraphicnamed{scalars//Gcup}
\begin{tikzpicture}[scale=0.6]
	\begin{pgfonlayer}{nodelayer}
		\node [style=none] (6) at (0.5, 0.25) {};
		\node [style=none] (7) at (-0.5, 0.25) {};
		\node [style=gn] (10) at (0, -0.25) {};
	\end{pgfonlayer}
	\begin{pgfonlayer}{edgelayer}
		\draw [in=150, out=-90] (7.center) to (10);
		\draw [in=-90, out=30] (10) to (6.center);
	\end{pgfonlayer}
\end{tikzpicture}
}
\endpgfgraphicnamed\quad \stackrel{\textup{\tiny CCC}}{=}\quad%
\beginpgfgraphicnamed{scalars//Gcup-0}
\InputIfFileExists{scalars//Gcup-0.tikz}{}{\input{./figures/scalars//Gcup-0.tikz}}
\endpgfgraphicnamed\quad \stackrel{\textup{\tiny (S1)}}{=}\quad%
\beginpgfgraphicnamed{scalars//Gcup-1}
\InputIfFileExists{scalars//Gcup-1.tikz}{}{\input{./figures/scalars//Gcup-1.tikz}}
\endpgfgraphicnamed\quad \stackrel{\textup{\tiny (S3$'$)}}{=}\quad%
\beginpgfgraphicnamed{scalars//Gcup-2}
\InputIfFileExists{scalars//Gcup-2.tikz}{}{\input{./figures/scalars//Gcup-2.tikz}}
\endpgfgraphicnamed\quad \stackrel{\textup{\tiny (S1)}}{=}\quad%
\beginpgfgraphicnamed{scalars//Gcup-3}
\InputIfFileExists{scalars//Gcup-3.tikz}{}{\input{./figures/scalars//Gcup-3.tikz}}
\endpgfgraphicnamed\quad \stackrel{\textup{\tiny (S1)}}{=}\quad%
\beginpgfgraphicnamed{scalars//Gcup-4}
\InputIfFileExists{scalars//Gcup-4.tikz}{}{\input{./figures/scalars//Gcup-4.tikz}}
\endpgfgraphicnamed\quad \stackrel{\textup{\tiny (\ref{S0})}}{=}\quad%
\beginpgfgraphicnamed{scalars//cup}
}
\endpgfgraphicnamed
 \end{equation*}
  Finally, 
  \begin{equation}\label{H2}
\beginpgfgraphicnamed{scalars//H-2-0}
\begin{tikzpicture}[scale=0.6]
	\begin{pgfonlayer}{nodelayer}
		\node [style=H box] (0) at (0, 0.35) {};
		\node [style=none] (2) at (0, 1) {};
		\node [style=none] (3) at (0, -1) {};
		\node [style=H box] (4) at (0, -0.35) {};
	\end{pgfonlayer}
	\begin{pgfonlayer}{edgelayer}
		\draw (2.center) to (0);
		\draw (0) to (4);
		\draw (4) to (3.center);
	\end{pgfonlayer}
\end{tikzpicture}
}
\endpgfgraphicnamed\quad \stackrel{\textup{\tiny (S2$'$)}}{=}\quad  %
\beginpgfgraphicnamed{scalars//H-2-1}
\InputIfFileExists{scalars//H-2-1.tikz}{}{\input{./figures/scalars//H-2-1.tikz}}
\endpgfgraphicnamed\quad \stackrel{\textup{\tiny (H)}}{=}\quad  %
\beginpgfgraphicnamed{scalars//H-2-2}
\begin{tikzpicture}[scale=0.6]
	\begin{pgfonlayer}{nodelayer}
		\node [style=none] (2) at (0, 1) {};
		\node [style=none] (3) at (0, -1) {};
		\node [style=gn] (4) at (0, 0) {};
	\end{pgfonlayer}
	\begin{pgfonlayer}{edgelayer}
		\draw (4) to (2.center);
		\draw (3.center) to (4);
	\end{pgfonlayer}
\end{tikzpicture}
}
\endpgfgraphicnamed\quad \stackrel{\textup{\tiny (\ref{S0})}}{=}\quad 
\beginpgfgraphicnamed{scalars//H-2-3}
\begin{tikzpicture}[scale=0.6]
	\begin{pgfonlayer}{nodelayer}
		\node [style=none] (2) at (0, 1) {};
		\node [style=none] (3) at (0, -1) {};
	\end{pgfonlayer}
	\begin{pgfonlayer}{edgelayer}
		\draw (2.center) to (3.center);
	\end{pgfonlayer}
\end{tikzpicture}
}
\endpgfgraphicnamed
  \end{equation}

\end{proof}

A first particular instance of the `only connectivity matters' meta rule is that H is self transpose, which can be derived in $\zxc$:

\begin{lem}\label{ruelht}
$\zxc \vdash~ %
\beginpgfgraphicnamed{scalars//H-trans-s}
\InputIfFileExists{scalars//H-trans-s.tikz}{}{\input{./figures/scalars//H-trans-s.tikz}}
\endpgfgraphicnamed~=~ %
\beginpgfgraphicnamed{scalars//H-trans-s-l}
\InputIfFileExists{scalars//H-trans-s-l.tikz}{}{\input{./figures/scalars//H-trans-s-l.tikz}}
\endpgfgraphicnamed$\quad \textup{(HT)}
\end{lem} 

\begin{proof}
 \begin{equation*} 
\beginpgfgraphicnamed{scalars//H-transp-1}
\begin{tikzpicture}[scale=0.6]
	\begin{pgfonlayer}{nodelayer}
		\node [style=none] (2) at (0, 1) {};
		\node [style=none] (3) at (0, -1) {};
		\node [style=H box] (4) at (0, 0) {};
	\end{pgfonlayer}
	\begin{pgfonlayer}{edgelayer}
		\draw (4) to (3.center);
		\draw (4) to (2.center);
	\end{pgfonlayer}
\end{tikzpicture}
}
\endpgfgraphicnamed\quad \stackrel{\textup{\tiny CCC}}{=}\quad  %
\beginpgfgraphicnamed{scalars//H-transp-2}
\InputIfFileExists{scalars//H-transp-2.tikz}{}{\input{./figures/scalars//H-transp-2.tikz}}
\endpgfgraphicnamed\quad \stackrel{\textup{\tiny (\ref{H2})}}{=}\quad  %
\beginpgfgraphicnamed{scalars//H-transp-3}
\InputIfFileExists{scalars//H-transp-3.tikz}{}{\input{./figures/scalars//H-transp-3.tikz}}
\endpgfgraphicnamed\quad \stackrel{\textup{\tiny (S3$'$)}}{=}\quad %
\beginpgfgraphicnamed{scalars//H-transp-4}
\InputIfFileExists{scalars//H-transp-4.tikz}{}{\input{./figures/scalars//H-transp-4.tikz}}
\endpgfgraphicnamed \quad \stackrel{\textup{\tiny (H)}}{=}\quad %
\beginpgfgraphicnamed{scalars//H-transp-5}
\InputIfFileExists{scalars//H-transp-5.tikz}{}{\input{./figures/scalars//H-transp-5.tikz}}
\endpgfgraphicnamed \quad \stackrel{\textup{\tiny (S3$'$)}}{=}\quad %
\beginpgfgraphicnamed{scalars//H-transp-6}
\InputIfFileExists{scalars//H-transp-6.tikz}{}{\input{./figures/scalars//H-transp-6.tikz}}
\endpgfgraphicnamed
   \end{equation*}
\end{proof}

Another instance of the `only connectivity matters' meta rule is the partial transpose of the green dot, which can also be derived in $\zxc$:
\begin{samepage}
\begin{lem}\label{lem:PT} ~
  \begin{center}
$\zxc \vdash~%
\beginpgfgraphicnamed{scalars-s//GRPT-L1}
\InputIfFileExists{scalars-s//GRPT-L1.tikz}{}{\input{./figures/scalars-s//GRPT-L1.tikz}}
\endpgfgraphicnamed ~~\stackrel{\textup{\tiny LPT}}{=}~~%
\beginpgfgraphicnamed{scalars-s//GRPT-C1}
\InputIfFileExists{scalars-s//GRPT-C1.tikz}{}{\input{./figures/scalars-s//GRPT-C1.tikz}}
\endpgfgraphicnamed~~\stackrel{\textup{\tiny RPT}}{=}~~ %
\beginpgfgraphicnamed{scalars-s//GRPT-R1}
\InputIfFileExists{scalars-s//GRPT-R1.tikz}{}{\input{./figures/scalars-s//GRPT-R1.tikz}}
\endpgfgraphicnamed$ and $\zxc \vdash~%
\beginpgfgraphicnamed{scalars-s//GRPT-L2}
\InputIfFileExists{scalars-s//GRPT-L2.tikz}{}{\input{./figures/scalars-s//GRPT-L2.tikz}}
\endpgfgraphicnamed ~~\stackrel{\textup{\tiny LPT}}{=}~~%
\beginpgfgraphicnamed{scalars-s//GRPT-C2}
\InputIfFileExists{scalars-s//GRPT-C2.tikz}{}{\input{./figures/scalars-s//GRPT-C2.tikz}}
\endpgfgraphicnamed~~\stackrel{\textup{\tiny RPT}}{=}~~ %
\beginpgfgraphicnamed{scalars-s//GRPT-R2}
\InputIfFileExists{scalars-s//GRPT-R2.tikz}{}{\input{./figures/scalars-s//GRPT-R2.tikz}}
\endpgfgraphicnamed \quad\textup{(PT)}$
  \end{center}
\end{lem} 
\end{samepage}

\begin{proof}Left partial transposes can be derived from (S3$'$) and (S1):
\[%
\beginpgfgraphicnamed{scalars-s//GLPT-1}
\InputIfFileExists{scalars-s//GLPT-1.tikz}{}{\input{./figures/scalars-s//GLPT-1.tikz}}
\endpgfgraphicnamed\quad \stackrel{\textup{\tiny (S3$'$)}}{=}\quad%
\beginpgfgraphicnamed{scalars-s//GLPT-2}
\InputIfFileExists{scalars-s//GLPT-2.tikz}{}{\input{./figures/scalars-s//GLPT-2.tikz}}
\endpgfgraphicnamed\quad \stackrel{\textup{\tiny (S1)}}{=}\quad%
\beginpgfgraphicnamed{scalars-s//GLPT-3}
\InputIfFileExists{scalars-s//GLPT-3.tikz}{}{\input{./figures/scalars-s//GLPT-3.tikz}}
\endpgfgraphicnamed\]

\noindent The right partial transpose can be derived as follows:
\[%
\beginpgfgraphicnamed{scalars-s//GRPT-1}
\InputIfFileExists{scalars-s//GRPT-1.tikz}{}{\input{./figures/scalars-s//GRPT-1.tikz}}
\endpgfgraphicnamed\quad \stackrel{\textup{\tiny (S1)}}{=}\quad%
\beginpgfgraphicnamed{scalars-s//GRPT-2b}
\InputIfFileExists{scalars-s//GRPT-2b.tikz}{}{\input{./figures/scalars-s//GRPT-2b.tikz}}
\endpgfgraphicnamed\quad \stackrel{\textup{\tiny (\ref{S0})}}{=}\quad%
\beginpgfgraphicnamed{scalars-s//GRPT-3b}
\InputIfFileExists{scalars-s//GRPT-3b.tikz}{}{\input{./figures/scalars-s//GRPT-3b.tikz}}
\endpgfgraphicnamed\quad \stackrel{\textup{\tiny CCC}}{=}\quad%
\beginpgfgraphicnamed{scalars-s//GRPT-8}
\InputIfFileExists{scalars-s//GRPT-8.tikz}{}{\input{./figures/scalars-s//GRPT-8.tikz}}
\endpgfgraphicnamed
\]

\noindent The derivation of the up-side-down versions of the partial transposes are similar. 
\end{proof}

A direct corollary of Lemma \ref{lem:PT} is the following alternative form of the spider rule:

\begin{cor}\label{cor:spider}
 $\zxc \vdash ~%
\beginpgfgraphicnamed{scalars//spider-ter-R}
\InputIfFileExists{scalars//spider-ter-R.tikz}{}{\input{./figures/scalars//spider-ter-R.tikz}}
\endpgfgraphicnamed~=~%
\beginpgfgraphicnamed{scalars//spider-ter-C}
\InputIfFileExists{scalars//spider-ter-C.tikz}{}{\input{./figures/scalars//spider-ter-C.tikz}}
\endpgfgraphicnamed
 $
\end{cor}

\begin{proof}
  \begin{align*}
\beginpgfgraphicnamed{scalars//spider-ter-R}
\InputIfFileExists{scalars//spider-ter-R.tikz}{}{\input{./figures/scalars//spider-ter-R.tikz}}
\endpgfgraphicnamed~&\stackrel{\textup{\tiny (S1)}}{=}~ %
\beginpgfgraphicnamed{scalars//spider-T-1}
\InputIfFileExists{scalars//spider-T-1.tikz}{}{\input{./figures/scalars//spider-T-1.tikz}}
\endpgfgraphicnamed~\stackrel{\textup{\tiny (S1)}}{=}~ %
\beginpgfgraphicnamed{scalars//spider-T-2}
\InputIfFileExists{scalars//spider-T-2.tikz}{}{\input{./figures/scalars//spider-T-2.tikz}}
\endpgfgraphicnamed~
    \\&
    \stackrel{\textup{\tiny (S3$'$), (\ref{S0})}}{=}~ %
\beginpgfgraphicnamed{scalars//spider-T-3}
\InputIfFileExists{scalars//spider-T-3.tikz}{}{\input{./figures/scalars//spider-T-3.tikz}}
\endpgfgraphicnamed~\stackrel{\textup{\tiny CCC}}{=}~%
\beginpgfgraphicnamed{scalars//spider-ter-L}
\InputIfFileExists{scalars//spider-ter-L.tikz}{}{\input{./figures/scalars//spider-ter-L.tikz}}
\endpgfgraphicnamed~\stackrel{\textup{\tiny (S1)}}{=}~%
\beginpgfgraphicnamed{scalars//spider-ter-C}
\InputIfFileExists{scalars//spider-ter-C.tikz}{}{\input{./figures/scalars//spider-ter-C.tikz}}
\endpgfgraphicnamed
  \end{align*}
\end{proof}

The most interesting instance of the `only connectivity matters' meta rule is the commutativity of the green dot, which  derivation in $\zxc$ is more involved:

\begin{lem}\label{hopfandcomm}
$\zxc \vdash~%
\beginpgfgraphicnamed{scalars-s//com-0}
\InputIfFileExists{scalars-s//com-0.tikz}{}{\input{./figures/scalars-s//com-0.tikz}}
\endpgfgraphicnamed =~%
\beginpgfgraphicnamed{scalars-s//com-R}
\InputIfFileExists{scalars-s//com-R.tikz}{}{\input{./figures/scalars-s//com-R.tikz}}
\endpgfgraphicnamed\quad \textup{(C)}$
\end{lem}

\begin{proof}
The derivation of (C) is based on the Hopf law $%
\beginpgfgraphicnamed{scalars//Hopf-L}
\InputIfFileExists{scalars//Hopf-L.tikz}{}{\input{./figures/scalars//Hopf-L.tikz}}
\endpgfgraphicnamed=%
\beginpgfgraphicnamed{scalars//Hopf-R}
\begin{tikzpicture}[scale=0.6]
	\begin{pgfonlayer}{nodelayer}
		\node [style=none] (0) at (0, -0.75) {};
		\node [style=none] (1) at (0, 0.75) {};
		\node [style=rn] (2) at (0, -0.25) {};
		\node [style=gn] (3) at (0, 0.25) {};
	\end{pgfonlayer}
	\begin{pgfonlayer}{edgelayer}
		\draw (0.center) to (2);
		\draw (3) to (1.center);
	\end{pgfonlayer}
\end{tikzpicture}}
\endpgfgraphicnamed$~. Using  $%
\beginpgfgraphicnamed{scalars//H2-small}
\begin{tikzpicture}[scale=0.6]
	\begin{pgfonlayer}{nodelayer}
		\node [style=none] (1) at (0, 0.5) {};
		\node [style=none] (2) at (0, -0.5) {};
		\node [style=H box] (3) at (0, -0.25) {};
		\node [style=H box] (4) at (0, 0.25) {};
	\end{pgfonlayer}
	\begin{pgfonlayer}{edgelayer}
		\draw (2.center) to (1.center);
	\end{pgfonlayer}
\end{tikzpicture}
}
\endpgfgraphicnamed= %
\beginpgfgraphicnamed{scalars//Id-2}
\begin{tikzpicture}[scale=0.6]
	\begin{pgfonlayer}{nodelayer}
		\node [style=none] (1) at (0, 0.5) {};
		\node [style=none] (2) at (0, -0.5) {};
	\end{pgfonlayer}
	\begin{pgfonlayer}{edgelayer}
		\draw (2.center) to (1.center);
	\end{pgfonlayer}
\end{tikzpicture}
}
\endpgfgraphicnamed$ (Lemma  \ref{ruels2}) and (H),  the Hopf law is equivalent to $%
\beginpgfgraphicnamed{scalars//HopfH-L}
\InputIfFileExists{scalars//HopfH-L.tikz}{}{\input{./figures/scalars//HopfH-L.tikz}}
\endpgfgraphicnamed=%
\beginpgfgraphicnamed{scalars//HopfH-R}
\begin{tikzpicture}[scale=0.6]
	\begin{pgfonlayer}{nodelayer}
		\node [style=none] (3) at (0, -0.75) {};
		\node [style=none] (13) at (0, 0.75) {};
		\node [style=gn] (14) at (0, -0.25) {};
		\node [style=gn] (18) at (0, 0.25) {};
	\end{pgfonlayer}
	\begin{pgfonlayer}{edgelayer}
		\draw (3.center) to (14.center);
		\draw (18) to (13.center);
	\end{pgfonlayer}
\end{tikzpicture}
}
\endpgfgraphicnamed$~ which can be derived in $\zxc$ as follows:
 
\begin{multline*}
\beginpgfgraphicnamed{scalars//HopfH-0}
\InputIfFileExists{scalars//HopfH-0.tikz}{}{\input{./figures/scalars//HopfH-0.tikz}}
\endpgfgraphicnamed ~\stackrel{\textup{\tiny \ref{ruels2}}}{=}~  
\beginpgfgraphicnamed{scalars//HopfH-0b}
\InputIfFileExists{scalars//HopfH-0b.tikz}{}{\input{./figures/scalars//HopfH-0b.tikz}}
\endpgfgraphicnamed ~\stackrel{\textup{\tiny (H)}}{=}~  
\beginpgfgraphicnamed{scalars//HopfH-1}
\InputIfFileExists{scalars//HopfH-1.tikz}{}{\input{./figures/scalars//HopfH-1.tikz}}
\endpgfgraphicnamed ~\stackrel{\textup{\tiny CCC}}{=}~  
\beginpgfgraphicnamed{scalars//HopfH-2}
\InputIfFileExists{scalars//HopfH-2.tikz}{}{\input{./figures/scalars//HopfH-2.tikz}}
\endpgfgraphicnamed ~\stackrel{\textup{\tiny *}}{=}~ 
\beginpgfgraphicnamed{scalars//HopfH-3}
\InputIfFileExists{scalars//HopfH-3.tikz}{}{\input{./figures/scalars//HopfH-3.tikz}}
\endpgfgraphicnamed ~\stackrel{\textup{\tiny \ref{ruels2}}}{=}~
\beginpgfgraphicnamed{scalars//HopfH-4b}
\InputIfFileExists{scalars//HopfH-4b.tikz}{}{\input{./figures/scalars//HopfH-4b.tikz}}
\endpgfgraphicnamed ~\stackrel{\textup{\tiny (S1)}}{=}~
\beginpgfgraphicnamed{scalars//HopfH-4}
\InputIfFileExists{scalars//HopfH-4.tikz}{}{\input{./figures/scalars//HopfH-4.tikz}}
\endpgfgraphicnamed \\ ~\stackrel{\textup{\tiny (B2$'$)}}{=}~ 
\beginpgfgraphicnamed{scalars//HopfH-5}
\InputIfFileExists{scalars//HopfH-5.tikz}{}{\input{./figures/scalars//HopfH-5.tikz}}
\endpgfgraphicnamed  ~\stackrel{\textup{\tiny (PT})}{=}~ 
\beginpgfgraphicnamed{scalars//HopfH-6}
\InputIfFileExists{scalars//HopfH-6.tikz}{}{\input{./figures/scalars//HopfH-6.tikz}}
\endpgfgraphicnamed ~\stackrel{\textup{\tiny (B1)}}{=}~ 
\beginpgfgraphicnamed{scalars//HopfH-7}
\InputIfFileExists{scalars//HopfH-7.tikz}{}{\input{./figures/scalars//HopfH-7.tikz}}
\endpgfgraphicnamed ~\stackrel{\textup{\tiny *}}{=}~
\beginpgfgraphicnamed{scalars//HopfH-8}
\InputIfFileExists{scalars//HopfH-8.tikz}{}{\input{./figures/scalars//HopfH-8.tikz}}
\endpgfgraphicnamed ~\stackrel{\textup{\tiny \ref{cor:spider}}}{=}~
\beginpgfgraphicnamed{scalars//HopfH-8b}
\InputIfFileExists{scalars//HopfH-8b.tikz}{}{\input{./figures/scalars//HopfH-8b.tikz}}
\endpgfgraphicnamed   ~\stackrel{\textup{\tiny \ref{ruels2}}}{=}~ 
\beginpgfgraphicnamed{scalars//HopfH-9}
\InputIfFileExists{scalars//HopfH-9.tikz}{}{\input{./figures/scalars//HopfH-9.tikz}}
\endpgfgraphicnamed  ~\stackrel{\textup{\tiny (H)}}{=}~ 
\beginpgfgraphicnamed{scalars//HopfH-10}
\InputIfFileExists{scalars//HopfH-10.tikz}{}{\input{./figures/scalars//HopfH-10.tikz}}
\endpgfgraphicnamed~\stackrel{\textup{\tiny \ref{cor:spider}}}{=}~ 
\beginpgfgraphicnamed{scalars//HopfH-11}
\begin{tikzpicture}[scale=0.6]
	\begin{pgfonlayer}{nodelayer}
		\node [style=gn] (2) at (0, 0.75) {};
		\node [style=none] (3) at (0, -1.25) {};
		\node [style=none] (13) at (0, 1.25) {};
		\node [style=gn] (14) at (0, -0.75) {};
	\end{pgfonlayer}
	\begin{pgfonlayer}{edgelayer}
		\draw (2.center) to (13.center);
		\draw (3.center) to (14.center);
	\end{pgfonlayer}
\end{tikzpicture}
}
\endpgfgraphicnamed
\end{multline*}

 \noindent where the second and seventh steps (*) are based on $%
\beginpgfgraphicnamed{scalars//RedCap-L}
\InputIfFileExists{scalars//RedCap-L.tikz}{}{\input{./figures/scalars//RedCap-L.tikz}}
\endpgfgraphicnamed=%
\beginpgfgraphicnamed{scalars//RedCap-R}
\begin{tikzpicture}[scale=0.6]
	\begin{pgfonlayer}{nodelayer}
		\node [style=none] (2) at (0.5, -0.325) {};
		\node [style=none] (3) at (-0.5, -0.325) {};
	\end{pgfonlayer}
	\begin{pgfonlayer}{edgelayer}
		\draw [in=90, out=90, looseness=2.00] (3.center) to (2.center);
	\end{pgfonlayer}
\end{tikzpicture}
}
\endpgfgraphicnamed$
 which can be derived as follows:
\begin{equation}\label{eq:red-cap}
\beginpgfgraphicnamed{scalars//RedCap-L}
\InputIfFileExists{scalars//RedCap-L.tikz}{}{\input{./figures/scalars//RedCap-L.tikz}}
\endpgfgraphicnamed ~\stackrel{\textup{\tiny (\ref{H2})}}{=}~ 
\beginpgfgraphicnamed{scalars//RedCap-1}
\InputIfFileExists{scalars//RedCap-1.tikz}{}{\input{./figures/scalars//RedCap-1.tikz}}
\endpgfgraphicnamed ~\stackrel{\textup{\tiny (H)}}{=}~ 
\beginpgfgraphicnamed{scalars//RedCap-2}
\InputIfFileExists{scalars//RedCap-2.tikz}{}{\input{./figures/scalars//RedCap-2.tikz}}
\endpgfgraphicnamed ~\stackrel{\textup{\tiny (S1)}}{=}~ 
\beginpgfgraphicnamed{scalars//RedCap-3}
\InputIfFileExists{scalars//RedCap-3.tikz}{}{\input{./figures/scalars//RedCap-3.tikz}}
\endpgfgraphicnamed ~\stackrel{\textup{\tiny (S3$'$)}}{=}~ 
\beginpgfgraphicnamed{scalars//RedCap-4}
\InputIfFileExists{scalars//RedCap-4.tikz}{}{\input{./figures/scalars//RedCap-4.tikz}}
\endpgfgraphicnamed ~\stackrel{\textup{\tiny (H)}}{=}~ 
\beginpgfgraphicnamed{scalars//RedCap-5}
\begin{tikzpicture}[scale=0.6]
	\begin{pgfonlayer}{nodelayer}
		\node [style=none] (2) at (0.5, -0.5) {};
		\node [style=none] (3) at (-0.5, -0.5) {};
		\node [style=gn] (4) at (0, 0.25) {};
	\end{pgfonlayer}
	\begin{pgfonlayer}{edgelayer}
		\draw [in=90, out=-135] (4) to (3.center);
		\draw [in=90, out=-45] (4) to (2.center);
	\end{pgfonlayer}
\end{tikzpicture}
}
\endpgfgraphicnamed ~\stackrel{\textup{\tiny (S3$'$)}}{=}~ 
\beginpgfgraphicnamed{scalars//RedCap-R}
\begin{tikzpicture}[scale=0.6]
	\begin{pgfonlayer}{nodelayer}
		\node [style=none] (2) at (0.5, -0.325) {};
		\node [style=none] (3) at (-0.5, -0.325) {};
	\end{pgfonlayer}
	\begin{pgfonlayer}{edgelayer}
		\draw [in=90, out=90, looseness=2.00] (3.center) to (2.center);
	\end{pgfonlayer}
\end{tikzpicture}
}
\endpgfgraphicnamed
\end{equation}

\noindent We are now ready to prove the commutativity property:\\
\centerline{$%
\beginpgfgraphicnamed{scalars-s//com-0}
\InputIfFileExists{scalars-s//com-0.tikz}{}{\input{./figures/scalars-s//com-0.tikz}}
\endpgfgraphicnamed~\stackrel{\textup{\tiny \ref{ruels2}}}{=}~ 
\beginpgfgraphicnamed{scalars-s//com-1}
\InputIfFileExists{scalars-s//com-1.tikz}{}{\input{./figures/scalars-s//com-1.tikz}}
\endpgfgraphicnamed~\stackrel{\textup{\tiny \ref{ruels2}}}{=}~ 
\beginpgfgraphicnamed{scalars-s//com-2}
\InputIfFileExists{scalars-s//com-2.tikz}{}{\input{./figures/scalars-s//com-2.tikz}}
\endpgfgraphicnamed~\stackrel{\textup{\tiny (S1),\ref{cor:spider}}}{=}~ 
\beginpgfgraphicnamed{scalars-s//com-3}
\InputIfFileExists{scalars-s//com-3.tikz}{}{\input{./figures/scalars-s//com-3.tikz}}
\endpgfgraphicnamed~\stackrel{\textup{\tiny Hopf}}{=}~ 
\beginpgfgraphicnamed{scalars-s//com-4}
\InputIfFileExists{scalars-s//com-4.tikz}{}{\input{./figures/scalars-s//com-4.tikz}}
\endpgfgraphicnamed~\stackrel{\textup{\tiny (S1),\ref{cor:spider}}}{=}~ 
\beginpgfgraphicnamed{scalars-s//com-5}
\InputIfFileExists{scalars-s//com-5.tikz}{}{\input{./figures/scalars-s//com-5.tikz}}
\endpgfgraphicnamed$}\\\centerline{$~\stackrel{\textup{\tiny (H)}}{=}~ 
\beginpgfgraphicnamed{scalars-s//com-6}
\InputIfFileExists{scalars-s//com-6.tikz}{}{\input{./figures/scalars-s//com-6.tikz}}
\endpgfgraphicnamed~\stackrel{\textup{\tiny (B2)}}{=}~ 
\beginpgfgraphicnamed{scalars-s//com-7}
\InputIfFileExists{scalars-s//com-7.tikz}{}{\input{./figures/scalars-s//com-7.tikz}}
\endpgfgraphicnamed~\stackrel{\textup{\tiny (S1),\ref{cor:spider}}}{=}~ 
\beginpgfgraphicnamed{scalars-s//com-8}
\InputIfFileExists{scalars-s//com-8.tikz}{}{\input{./figures/scalars-s//com-8.tikz}}
\endpgfgraphicnamed~\stackrel{\textup{\tiny Hopf}}{=}~ 
\beginpgfgraphicnamed{scalars-s//com-9}
\InputIfFileExists{scalars-s//com-9.tikz}{}{\input{./figures/scalars-s//com-9.tikz}}
\endpgfgraphicnamed~\stackrel{\textup{\tiny (B1)}}{=}~ 
\beginpgfgraphicnamed{scalars-s//com-10}
\InputIfFileExists{scalars-s//com-10.tikz}{}{\input{./figures/scalars-s//com-10.tikz}}
\endpgfgraphicnamed$}\\\centerline{$~\stackrel{\textup{\tiny (H)}}{=}~ 
\beginpgfgraphicnamed{scalars-s//com-11}
\InputIfFileExists{scalars-s//com-11.tikz}{}{\input{./figures/scalars-s//com-11.tikz}}
\endpgfgraphicnamed~\stackrel{\textup{\tiny \ref{cor:spider}}}{=}~ 
\beginpgfgraphicnamed{scalars-s//com-12}
\InputIfFileExists{scalars-s//com-12.tikz}{}{\input{./figures/scalars-s//com-12.tikz}}
\endpgfgraphicnamed~\stackrel{\textup{\tiny \ref{ruels2}}}{=}~ 
\beginpgfgraphicnamed{scalars-s//com-R}
\InputIfFileExists{scalars-s//com-R.tikz}{}{\input{./figures/scalars-s//com-R.tikz}}
\endpgfgraphicnamed$}

\end{proof}

\begin{thm}\label{thm:zxc-connectivity}
$\zxc$ satisfies the `only connectivity matters' meta rule. As a consequence, $\zxc$ is complete for stabilizer quantum mechanics. 
\end{thm}

\begin{proof}
First notice that the upside-down versions of the equation (C) can  be derived:
\[%
\beginpgfgraphicnamed{scalars-s//com-USD-1}
\InputIfFileExists{scalars-s//com-USD-1.tikz}{}{\input{./figures/scalars-s//com-USD-1.tikz}}
\endpgfgraphicnamed\quad \stackrel{\textup{\tiny (PT)}}{=}\quad%
\beginpgfgraphicnamed{scalars-s//com-USD-2}
\InputIfFileExists{scalars-s//com-USD-2.tikz}{}{\input{./figures/scalars-s//com-USD-2.tikz}}
\endpgfgraphicnamed\quad\stackrel{\textup{\tiny CCC}}{=}\quad%
\beginpgfgraphicnamed{scalars-s//com-USD-3}
\InputIfFileExists{scalars-s//com-USD-3.tikz}{}{\input{./figures/scalars-s//com-USD-3.tikz}}
\endpgfgraphicnamed\quad\stackrel{\textup{\tiny (C)}}{=}\quad%
\beginpgfgraphicnamed{scalars-s//com-USD-4}
\InputIfFileExists{scalars-s//com-USD-4.tikz}{}{\input{./figures/scalars-s//com-USD-4.tikz}}
\endpgfgraphicnamed\quad\stackrel{\textup{\tiny (PT)}}{=}\quad%
\beginpgfgraphicnamed{scalars-s//commute1}
\InputIfFileExists{scalars-s//commute1.tikz}{}{\input{./figures/scalars-s//commute1.tikz}}
\endpgfgraphicnamed\]
Moreover, commutativity can also be derived for spiders of arbitrary degree: 
\[%
\beginpgfgraphicnamed{scalars-s//GCom-1}
\InputIfFileExists{scalars-s//GCom-1.tikz}{}{\input{./figures/scalars-s//GCom-1.tikz}}
\endpgfgraphicnamed\quad \stackrel{\textup{\tiny (S1)}}{=}\quad%
\beginpgfgraphicnamed{scalars-s//GCom-2b}
\InputIfFileExists{scalars-s//GCom-2b.tikz}{}{\input{./figures/scalars-s//GCom-2b.tikz}}
\endpgfgraphicnamed\quad \stackrel{\textup{\tiny (PT)}}{=}\quad%
\beginpgfgraphicnamed{scalars-s//GCom-3}
\InputIfFileExists{scalars-s//GCom-3.tikz}{}{\input{./figures/scalars-s//GCom-3.tikz}}
\endpgfgraphicnamed\quad \stackrel{\textup{\tiny (S1)}}{=}\quad%
\beginpgfgraphicnamed{scalars-s//GCom-4}
\InputIfFileExists{scalars-s//GCom-4.tikz}{}{\input{./figures/scalars-s//GCom-4.tikz}}
\endpgfgraphicnamed\]\[ \stackrel{\textup{\tiny (S1)}}{=}\quad%
\beginpgfgraphicnamed{scalars-s//GCom-5}
\InputIfFileExists{scalars-s//GCom-5.tikz}{}{\input{./figures/scalars-s//GCom-5.tikz}}
\endpgfgraphicnamed
\quad \stackrel{\textup{\tiny (C)}}{=}\quad%
\beginpgfgraphicnamed{scalars-s//GCom-6}
\InputIfFileExists{scalars-s//GCom-6.tikz}{}{\input{./figures/scalars-s//GCom-6.tikz}}
\endpgfgraphicnamed\quad \stackrel{\textup{\tiny (S1)}}{=}\quad%
\beginpgfgraphicnamed{scalars-s//GCom-7}
\InputIfFileExists{scalars-s//GCom-7.tikz}{}{\input{./figures/scalars-s//GCom-7.tikz}}
\endpgfgraphicnamed\quad \stackrel{\textup{\tiny (S1)}}{=}\quad%
\beginpgfgraphicnamed{scalars-s//GCom-8}
\InputIfFileExists{scalars-s//GCom-8.tikz}{}{\input{./figures/scalars-s//GCom-8.tikz}}
\endpgfgraphicnamed\]\[ \stackrel{\textup{\tiny (PT)}}{=}\quad%
\beginpgfgraphicnamed{scalars-s//GCom-9}
\InputIfFileExists{scalars-s//GCom-9.tikz}{}{\input{./figures/scalars-s//GCom-9.tikz}}
\endpgfgraphicnamed\quad \stackrel{\textup{\tiny (S1)}}{=}\quad%
\beginpgfgraphicnamed{scalars-s//GCom-10}
\InputIfFileExists{scalars-s//GCom-10.tikz}{}{\input{./figures/scalars-s//GCom-10.tikz}}
\endpgfgraphicnamed
\]

\noindent Similarly, 
\[%
\beginpgfgraphicnamed{scalars-s//GCom-USD-L}
\InputIfFileExists{scalars-s//GCom-USD-L.tikz}{}{\input{./figures/scalars-s//GCom-USD-L.tikz}}
\endpgfgraphicnamed\quad \stackrel{\textup{\tiny }}{=}\quad%
\beginpgfgraphicnamed{scalars-s//GCom-USD-R}
\InputIfFileExists{scalars-s//GCom-USD-R.tikz}{}{\input{./figures/scalars-s//GCom-USD-R.tikz}}
\endpgfgraphicnamed\]

So far, we have proved all the required properties of the green spiders, which means two green spiders with the same phase are equal if and only if they have the same numbers of inputs and outputs. The same result holds for red spiders, since
the colour swapped versions of the previous equations can be derived thanks to the (H) rule and Equation \ref{H2}.

Now we prove that $\zxc$ satisfies the `only connectivity matters' meta rule in terms of Definition~\ref{dfn:connectivity}. 
An explicit application of the proof is demonstrated in Example~\ref{ex:connectivity} below.

Suppose $D_1$ and $D_2$ are two \ZX-calculus diagrams in $\zxc$. 
Let $G_{D_1}=(V_1,E_1,\ell_1)$ and $G_{D_2}=(V_2,E_2,\ell_2)$ be the corresponding labelled multigraphs. Assume that there exists an isomorphism $h$ from $G_{D_1}$ to $G_{D_2}$ which respects the labelling.
We want to show that $\zxc$ can be used to transform $D_1$ into $D_2$.

To do this, we first transform $D_1$ into another diagram $D'_1$ by locally modifying each node using (HT), (PT), and commutativity of spiders.
In particular, we replace each spider $u$ by a spider $u'$ of the same colour and angle as $u$, such that if the $k$-th input (or output) of $h(u)$  is connected to $h(v)$ in $D_2$, then the $k$-th input (or output) of $u'$ is connected to $v'$ in $D'_1$. 
Similarly, we replace each Hadamard node $w$ by a Hadamard node $w'$ such that if the input of $h(w)$  is connected to $h(v)$ in $D_2$, then the input of $w'$ is connected to $v'$ in $D'_1$, and similarly for the output of the node.
This replacement will generally introduce new cups, caps and swaps in the neighbourhoods of the nodes.
Since $\zxc$ implies (HT), (PT), and commutativity of spiders, we have $\zxc \vdash D_1=D_1'$.

Now, by construction, the diagram $D_1'$ is isomorphic to $D_2$ in the sense of Theorem \ref{thm:compact-closed} \cite[Theorem~14]{selinger_survey_2010}, respecting the order and direction of incidence of wires on nodes. Therefore, since $\zxc$ includes the axioms of a compact closed category, $\zxc \vdash D'_1=D_2$.
By combining this derivation with the previous one, $\zxc \vdash D_1=D_2$.
Thus, $\zxc$ satisfies the `only connectivity matters' meta rule.
\end{proof}
 
\begin{exa}\label{ex:connectivity}
 To illustrate the final part of Theorem~\ref{thm:zxc-connectivity}, let $D_1$ and $D_2$ be the following two diagrams:
\[D_1 ~:= ~%
\beginpgfgraphicnamed{ex_OTM2}
\InputIfFileExists{ex_OTM2.tikz}{}{\input{./figures/ex_OTM2.tikz}}
\endpgfgraphicnamed \qquad\qquad D_2 ~:=~ %
\beginpgfgraphicnamed{ex_OTM1}
\InputIfFileExists{ex_OTM1.tikz}{}{\input{./figures/ex_OTM1.tikz}}
\endpgfgraphicnamed\]

The isomorphism of labelled multigraphs $h$ simply maps the green spider in $D_1$ to the green spider in $D_2$ and the red spider in $D_1$ to the red spider in $D_2$. 
The intermediate diagram $D'_1$ is the following: 

\[D_1'~:=~ %
\beginpgfgraphicnamed{ex_OTM2-primeb}
\InputIfFileExists{ex_OTM2-primeb.tikz}{}{\input{./figures/ex_OTM2-primeb.tikz}}
\endpgfgraphicnamed\]
 
$D_1'$ is obtained from $D_1$ by transforming the generators of $D_1$ locally through applying the rules (HT), (PT), and commutativity of spiders inside the (informal) blue boxes:
\[\zxc \vdash %
\beginpgfgraphicnamed{ex_OTM2-1}
\InputIfFileExists{ex_OTM2-1.tikz}{}{\input{./figures/ex_OTM2-1.tikz}}
\endpgfgraphicnamed ~~=~~ %
\beginpgfgraphicnamed{ex_OTM2-2b}
\InputIfFileExists{ex_OTM2-2b.tikz}{}{\input{./figures/ex_OTM2-2b.tikz}}
\endpgfgraphicnamed\]

Now $D_1'$ and $D_2$ are isomorphic in the sense of Theorem~\ref{thm:compact-closed}, i.e.\ one can move around the generators (depicted with black boxes) to transform $D_1'$ into $D_2$ without changing the inputs and outputs of the spiders inside the boxes:

\[\zxc \vdash%
\beginpgfgraphicnamed{ex_OTM2-4b}
\InputIfFileExists{ex_OTM2-4b.tikz}{}{\input{./figures/ex_OTM2-4b.tikz}}
\endpgfgraphicnamed ~~=~~ %
\beginpgfgraphicnamed{ex_OTM3-1}
\InputIfFileExists{ex_OTM3-1.tikz}{}{\input{./figures/ex_OTM3-1.tikz}}
\endpgfgraphicnamed\]

Thus, the equality between $D_1'$ and $D_2$ follows from the axioms of a compact closed category, completing the example.
\end{exa}

We have derived the `only connectivity matters' meta rule from the ambient compact closed category and the rules of Figure \ref{figure1} together with the additional (S2$'$) rule $%
\beginpgfgraphicnamed{scalars//RedID-2}
\begin{tikzpicture}[scale=0.6]
	\begin{pgfonlayer}{nodelayer}
		\node [style=rn] (0) at (0, 0) {};
		\node [style=none] (1) at (0, 0.5) {};
		\node [style=none] (2) at (0, -0.5) {};
	\end{pgfonlayer}
	\begin{pgfonlayer}{edgelayer}
		\draw (0) to (2.center);
		\draw (0) to (1.center);
	\end{pgfonlayer}
\end{tikzpicture}
}
\endpgfgraphicnamed~=~ %
\beginpgfgraphicnamed{scalars//Id-2}
\begin{tikzpicture}[scale=0.6]
	\begin{pgfonlayer}{nodelayer}
		\node [style=none] (1) at (0, 0.5) {};
		\node [style=none] (2) at (0, -0.5) {};
	\end{pgfonlayer}
	\begin{pgfonlayer}{edgelayer}
		\draw (2.center) to (1.center);
	\end{pgfonlayer}
\end{tikzpicture}
}
\endpgfgraphicnamed$~. While (S2$'$) can be derived from (S3$'$) in $\zxs$, we conjecture that (S2$'$) is necessary in $\zxc$, i.e. $\zxc \setminus $(S2$'$)$ \not \vdash $(S2$'$), although we do not have a proof of this.

\subsection{Braided autonomous category / 3D isotopy}
\label{s:braided_category}

In this subsection, we take a less standard approach for making the connectivity meta rule rigorous: we work in an ambient category which implies only that diagrams which are 3D-isotopic are equal (whereas a compact closed category implies that all isomorphic diagrams are equal).
We show that, combined with the other rules of the \ZX-calculus, 3D-isotopy is enough to recover the `only the connectivity matters' meta-rule.

3D-isotopy  is a natural equivalence of diagrams which can be axiomatised using the Reidemeister moves \cite{reidemeister_knotentheorie_1932} (see Figure \ref{fig:reidrules}), the snake equations (see Figure \ref{fig:compact_structure}, not including the equations where caps and cups are symmetrical), as well as the property that arbitrary maps can slide freely along either wire in a braiding. In a categorical setting, the Reidemeister move (R2) follows from the invertibility of a braiding, while (R3) follows from the coherence axioms of a braided monoidal category and the naturality of a braiding. Therefore,  3D-isotopy is modelled by a braided autonomous category augmented with the loop axiom (R1) \cite{selinger_survey_2010}, which appears so useful that it is exploited by several graphical languages for quantum information and computation \cite{reutter2019shaded, jaffe2018holographic}.

\begin{figure}[h!]
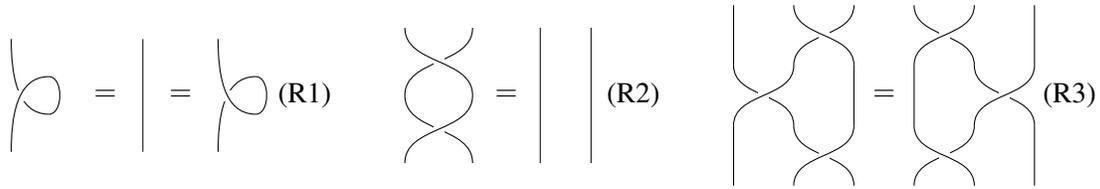

 \centering%
\beginpgfgraphicnamed{scalars//reid1b}
\InputIfFileExists{scalars//reid1b.tikz}{}{\input{./figures/scalars//reid1b.tikz}}
\endpgfgraphicnamed(R1) \qquad %
\beginpgfgraphicnamed{scalars//reid2b}
\InputIfFileExists{scalars//reid2b.tikz}{}{\input{./figures/scalars//reid2b.tikz}}
\endpgfgraphicnamed (R2) \qquad %
\beginpgfgraphicnamed{scalars//reid3c}
\InputIfFileExists{scalars//reid3c.tikz}{}{\input{./figures/scalars//reid3c.tikz}}
\endpgfgraphicnamed(R3)
 \caption{Reidemeister moves}\label{fig:reidrules}
\end{figure}

The following technicality arises: when the ambient category is braided rather than symmetric, one needs to specify which way the wires cross in each crossing. The only crossing occurring in the rules of the language is in the bialgebra rule (B2$'$), which we transform into the following braided rule (the choice of how the wires cross is arbitrary):
\begin{equation}\tag{B2$''$}\label{eq:B2''}
\beginpgfgraphicnamed{scalars//b2snewbraid1b}
\InputIfFileExists{scalars//b2snewbraid1b.tikz}{}{\input{./figures/scalars//b2snewbraid1b.tikz}}
\endpgfgraphicnamed
\end{equation}

Now we define a version of the \ZX-calculus based on a braided autonomous category, called $\zxb$.

\begin{defi}
 The graphical calculus $\zxb$ has the same generators as $\zxs$ (cf.\ the table at the beginning of Section~\ref{s:generators}), except that the swap $\sigma$ is replaced by the two braidings
 \[
  \beta: 2\to 2 :: %
\beginpgfgraphicnamed{scalars//braid1}
\InputIfFileExists{scalars//braid1.tikz}{}{\input{./figures/scalars//braid1.tikz}}
\endpgfgraphicnamed \qquad\text{and}\qquad
  \beta^{-1}: 2\to 2 :: %
\beginpgfgraphicnamed{scalars//braid2}
\InputIfFileExists{scalars//braid2.tikz}{}{\input{./figures/scalars//braid2.tikz}}
\endpgfgraphicnamed
 \]
 The graphical rewrite rules of $\zxb$ are those given in Figure \ref{figure1}, with (B2$'$) replaced by \eqref{eq:B2''}, together with 3D-isotopy and the rule \eqref{eq:s2prime}.
\end{defi}

We also note here that the ellipsis notations ($\ldots$) in the rules of $\zxb$  do not involve any crossing of wires. 

Below we will show that the $\zxb$ is complete for stabilizer quantum mechanics. To achieve this goal we need a series of lemmas. 

First note that in $\zxb$ we still have the equalities derived in Lemmas \ref{ruels2}, \ref{ruelht},  \ref{lem:PT}  and Corollary \ref{cor:spider},  since all the rules applied there  exist in the $\zxb$ as well.  This includes the equality $%
\beginpgfgraphicnamed{scalars//H2-small}
\begin{tikzpicture}[scale=0.6]
	\begin{pgfonlayer}{nodelayer}
		\node [style=none] (1) at (0, 0.5) {};
		\node [style=none] (2) at (0, -0.5) {};
		\node [style=H box] (3) at (0, -0.25) {};
		\node [style=H box] (4) at (0, 0.25) {};
	\end{pgfonlayer}
	\begin{pgfonlayer}{edgelayer}
		\draw (2.center) to (1.center);
	\end{pgfonlayer}
\end{tikzpicture}
}
\endpgfgraphicnamed= %
\beginpgfgraphicnamed{scalars//Id-2}
\begin{tikzpicture}[scale=0.6]
	\begin{pgfonlayer}{nodelayer}
		\node [style=none] (1) at (0, 0.5) {};
		\node [style=none] (2) at (0, -0.5) {};
	\end{pgfonlayer}
	\begin{pgfonlayer}{edgelayer}
		\draw (2.center) to (1.center);
	\end{pgfonlayer}
\end{tikzpicture}
}
\endpgfgraphicnamed$ (Lemma  \ref{ruels2}), together with (H), from which it follows that the
 colour swapped versions of all the rules and their derivations in $\zxb$  still hold. In particular, we have red spiders, partial transpose of the red dot, and  a braided version of (B2) and its colour-swapped version:
 \begin{equation}\label{eq:b2braid}
\beginpgfgraphicnamed{scalars//b2braidscirc}
\InputIfFileExists{scalars//b2braidscirc.tikz}{}{\input{./figures/scalars//b2braidscirc.tikz}}
\endpgfgraphicnamed
\end{equation}

We also have the Hopf law.

\begin{lem}\label{hopfbraid} The Hopf law holds in $\zxb$:
 $$%
\beginpgfgraphicnamed{scalars//Hopf-L-flip-flip}
\InputIfFileExists{scalars//Hopf-L-flip-flip.tikz}{}{\input{./figures/scalars//Hopf-L-flip-flip.tikz}}
\endpgfgraphicnamed=%
\beginpgfgraphicnamed{scalars//Hopf-R-flip-flip}
\begin{tikzpicture}[yscale=-1]
	\begin{pgfonlayer}{nodelayer}
		\node [style=gn] (0) at (0, -0.25) {};
		\node [style=none] (1) at (0, -0.5) {};
		\node [style=rn] (2) at (0, 0.25) {};
		\node [style=none] (3) at (0, 0.5) {};
	\end{pgfonlayer}
	\begin{pgfonlayer}{edgelayer}
		\draw (3.center) to (2);
		\draw (0) to (1.center);
	\end{pgfonlayer}
\end{tikzpicture}
}
\endpgfgraphicnamed$$ 
\end{lem}

\begin{proof}
We first prove the upside-down Hopf law $%
\beginpgfgraphicnamed{scalars//Hopf-L-flip}
\InputIfFileExists{scalars//Hopf-L-flip.tikz}{}{\input{./figures/scalars//Hopf-L-flip.tikz}}
\endpgfgraphicnamed=%
\beginpgfgraphicnamed{scalars//Hopf-R-flip}
\begin{tikzpicture}[scale=0.6,yscale=-1]
	\begin{pgfonlayer}{nodelayer}
		\node [style=none] (0) at (0, -0.75) {};
		\node [style=none] (1) at (0, 0.75) {};
		\node [style=rn] (2) at (0, -0.25) {};
		\node [style=gn] (3) at (0, 0.25) {};
	\end{pgfonlayer}
	\begin{pgfonlayer}{edgelayer}
		\draw (0.center) to (2);
		\draw (3) to (1.center);
	\end{pgfonlayer}
\end{tikzpicture}
}
\endpgfgraphicnamed$~. Then the normal Hopf law follows directly from  $%
\beginpgfgraphicnamed{scalars//H2-small}
\begin{tikzpicture}[scale=0.6]
	\begin{pgfonlayer}{nodelayer}
		\node [style=none] (1) at (0, 0.5) {};
		\node [style=none] (2) at (0, -0.5) {};
		\node [style=H box] (3) at (0, -0.25) {};
		\node [style=H box] (4) at (0, 0.25) {};
	\end{pgfonlayer}
	\begin{pgfonlayer}{edgelayer}
		\draw (2.center) to (1.center);
	\end{pgfonlayer}
\end{tikzpicture}
}
\endpgfgraphicnamed= %
\beginpgfgraphicnamed{scalars//Id-2}
\begin{tikzpicture}[scale=0.6]
	\begin{pgfonlayer}{nodelayer}
		\node [style=none] (1) at (0, 0.5) {};
		\node [style=none] (2) at (0, -0.5) {};
	\end{pgfonlayer}
	\begin{pgfonlayer}{edgelayer}
		\draw (2.center) to (1.center);
	\end{pgfonlayer}
\end{tikzpicture}
}
\endpgfgraphicnamed$ (Lemma  \ref{ruels2}) and (H).  Indeed,

\noindent  $ 
\beginpgfgraphicnamed{scalars//HopfH-1-v2}
\InputIfFileExists{scalars//HopfH-1-v2.tikz}{}{\input{./figures/scalars//HopfH-1-v2.tikz}}
\endpgfgraphicnamed ~\stackrel{\textup{\tiny (R1)}}{=}~  
\beginpgfgraphicnamed{scalars//HopfH-2-v2}
\InputIfFileExists{scalars//HopfH-2-v2.tikz}{}{\input{./figures/scalars//HopfH-2-v2.tikz}}
\endpgfgraphicnamed ~\stackrel{\textup{\tiny *}}{=}~ 
\beginpgfgraphicnamed{scalars//HopfH-3-v2}
\InputIfFileExists{scalars//HopfH-3-v2.tikz}{}{\input{./figures/scalars//HopfH-3-v2.tikz}}
\endpgfgraphicnamed ~\stackrel{\textup{\tiny \ref{ruels2}}}{=}~
\beginpgfgraphicnamed{scalars//HopfH-4b-v2}
\InputIfFileExists{scalars//HopfH-4b-v2.tikz}{}{\input{./figures/scalars//HopfH-4b-v2.tikz}}
\endpgfgraphicnamed ~\stackrel{\textup{\tiny (S1)}}{=}~
\beginpgfgraphicnamed{scalars//HopfH-4-v2}
\InputIfFileExists{scalars//HopfH-4-v2.tikz}{}{\input{./figures/scalars//HopfH-4-v2.tikz}}
\endpgfgraphicnamed  ~\stackrel{\textup{\tiny (B2$''$)}}{=}~ 
\beginpgfgraphicnamed{scalars//HopfH-5-v2}
\InputIfFileExists{scalars//HopfH-5-v2.tikz}{}{\input{./figures/scalars//HopfH-5-v2.tikz}}
\endpgfgraphicnamed $\\$ ~\stackrel{\textup{\tiny (PT)}}{=}~ 
\beginpgfgraphicnamed{scalars//HopfH-6-v2}
\InputIfFileExists{scalars//HopfH-6-v2.tikz}{}{\input{./figures/scalars//HopfH-6-v2.tikz}}
\endpgfgraphicnamed ~\stackrel{\textup{\tiny (B1)}}{=}~ 
\beginpgfgraphicnamed{scalars//HopfH-7-v2}
\InputIfFileExists{scalars//HopfH-7-v2.tikz}{}{\input{./figures/scalars//HopfH-7-v2.tikz}}
\endpgfgraphicnamed ~\stackrel{\textup{\tiny *}}{=}~
\beginpgfgraphicnamed{scalars//HopfH-8-v2}
\InputIfFileExists{scalars//HopfH-8-v2.tikz}{}{\input{./figures/scalars//HopfH-8-v2.tikz}}
\endpgfgraphicnamed ~\stackrel{\textup{\tiny (PT)}}{=}~
\beginpgfgraphicnamed{scalars//Hopf-R-flip}
\begin{tikzpicture}[scale=0.6,yscale=-1]
	\begin{pgfonlayer}{nodelayer}
		\node [style=none] (0) at (0, -0.75) {};
		\node [style=none] (1) at (0, 0.75) {};
		\node [style=rn] (2) at (0, -0.25) {};
		\node [style=gn] (3) at (0, 0.25) {};
	\end{pgfonlayer}
	\begin{pgfonlayer}{edgelayer}
		\draw (0.center) to (2);
		\draw (3) to (1.center);
	\end{pgfonlayer}
\end{tikzpicture}
}
\endpgfgraphicnamed
 $
 
 \noindent where the second and eighth steps (*) are based on the identity $%
\beginpgfgraphicnamed{scalars//RedCap-L}
\InputIfFileExists{scalars//RedCap-L.tikz}{}{\input{./figures/scalars//RedCap-L.tikz}}
\endpgfgraphicnamed=%
\beginpgfgraphicnamed{scalars//RedCap-R}
\begin{tikzpicture}[scale=0.6]
	\begin{pgfonlayer}{nodelayer}
		\node [style=none] (2) at (0.5, -0.325) {};
		\node [style=none] (3) at (-0.5, -0.325) {};
	\end{pgfonlayer}
	\begin{pgfonlayer}{edgelayer}
		\draw [in=90, out=90, looseness=2.00] (3.center) to (2.center);
	\end{pgfonlayer}
\end{tikzpicture}
}
\endpgfgraphicnamed$, which is proved in \eqref{eq:red-cap}.
 The derivation of this identity goes through the same way in $\zxb$.
 \end{proof}

Now we can prove the braided commutativity of green co-copy:
\begin{lem}\label{greencommute2}
The green co-copy map is braided commutative:
\begin{equation}\label{eq:greencommutebraid}
\beginpgfgraphicnamed{scalars//greencommutebraid}
\InputIfFileExists{scalars//greencommutebraid.tikz}{}{\input{./figures/scalars//greencommutebraid.tikz}}
\endpgfgraphicnamed
 \end{equation}
\end{lem}

\begin{proof}
The obvious rewrite rule for removing a wire crossing is the braided (B2), i.e, (\ref{eq:b2braid}). We rewrite the diagram so that can be applied, using (S2$'$), the spider rules, and the Hopf law (which is used twice, symmetrically). This covers the rewrite steps in the top row. (\ref{eq:b2braid}) is applied over the line break.
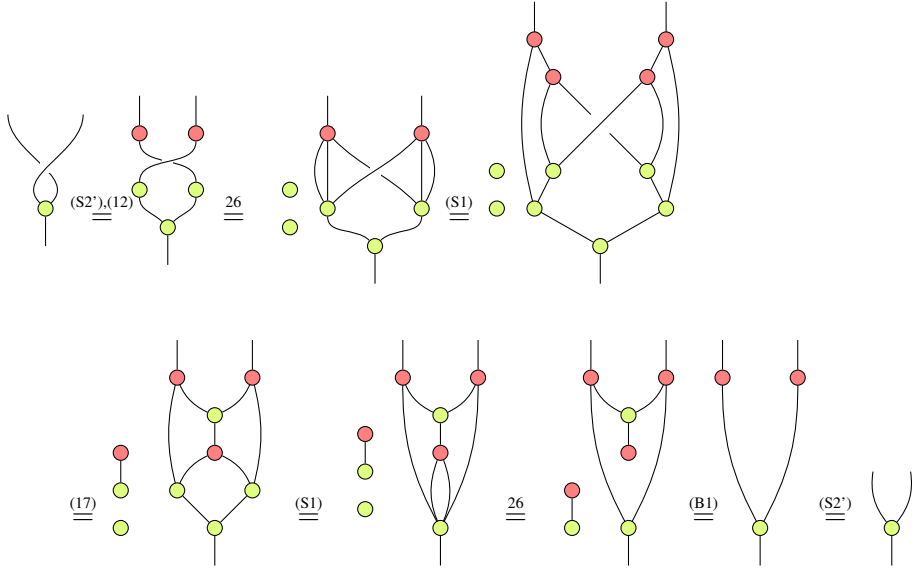
\begin{align*}
\beginpgfgraphicnamed{scalars//greencommuteprfbraid_v3}
\InputIfFileExists{scalars//greencommuteprfbraid_v3.tikz}{}{\input{./figures/scalars//greencommuteprfbraid_v3.tikz}}
\endpgfgraphicnamed
\end{align*}
We then use the spider rule, the Hopf law again, the upside-down copy law obtained by partial transpose from (B1), and (S2$'$) to simplify the diagram again, thus completing the proof.
\end{proof}

\begin{lem}\label{greencommuteflip}
The green copy map is braided commutative:
\begin{equation}\label{eq:greencommuteflip}
\beginpgfgraphicnamed{scalars//greencommuteflipbraid}
\InputIfFileExists{scalars//greencommuteflipbraid.tikz}{}{\input{./figures/scalars//greencommuteflipbraid.tikz}}
\endpgfgraphicnamed
 \end{equation}
\end{lem}

\begin{proof}
We have upside-down versions of all the rules used in the proof of that the green co-copy map is braided commutative (Lemma \ref{greencommute2}). That proof can therefore be straightforwardly repeated upside-down.
\end{proof}

The colour-swapped versions of the above Lemmas also hold:

\begin{lem}\label{redcommuteflipswap}
Both the red copy and co-copy maps are braided commutative:
\begin{equation}
\beginpgfgraphicnamed{scalars//redcommutebraid}
\InputIfFileExists{scalars//redcommutebraid.tikz}{}{\input{./figures/scalars//redcommutebraid.tikz}}
\endpgfgraphicnamed\quad\quad\quad\quad%
\beginpgfgraphicnamed{scalars//redcommuteflipbraid}
\InputIfFileExists{scalars//redcommuteflipbraid.tikz}{}{\input{./figures/scalars//redcommuteflipbraid.tikz}}
\endpgfgraphicnamed
\end{equation}
\end{lem}

 \begin{proof}
These follow immediately from applying Hadamard nodes to all inputs and outputs of Lemmas \ref{greencommute2} and \ref{greencommuteflip} by the colour-swapped version of the colour change rule (H) and the naturality of the braiding.
\end{proof}

Once we have the braided commutativity of green co-copy, the inversely braided commutativity can be obtained immediately:

\begin{lem}\label{greeninvcommute}
The green co-copy map is inversely braided commutative:
\begin{equation}
\beginpgfgraphicnamed{scalars//greencommutebraid2}
\InputIfFileExists{scalars//greencommutebraid2.tikz}{}{\input{./figures/scalars//greencommutebraid2.tikz}}
\endpgfgraphicnamed
 \end{equation}
\end{lem}
\begin{proof}
\begin{align*}
\beginpgfgraphicnamed{scalars//greencommutebraid2prf}
\InputIfFileExists{scalars//greencommutebraid2prf.tikz}{}{\input{./figures/scalars//greencommutebraid2prf.tikz}}
\endpgfgraphicnamed
\end{align*}
Here we used the inverse property of the braiding and Lemma \ref{greencommute2}.
\end{proof}
As a consequence, we have inversely  braided commutativity of green copy, red copy and co-copy.
\begin{lem}\label{greenredcommuteflipswap}
The maps of green copy, red copy and co-copy are inversely braided commutative:
\begin{equation}\label{greenredcommuteflipswapbraid}
\beginpgfgraphicnamed{scalars//greencommuteflipbraid2}
\InputIfFileExists{scalars//greencommuteflipbraid2.tikz}{}{\input{./figures/scalars//greencommuteflipbraid2.tikz}}
\endpgfgraphicnamed\quad\quad\quad\quad%
\beginpgfgraphicnamed{scalars//redcommuteflipbraid2}
\InputIfFileExists{scalars//redcommuteflipbraid2.tikz}{}{\input{./figures/scalars//redcommuteflipbraid2.tikz}}
\endpgfgraphicnamed\quad\quad\quad\quad%
\beginpgfgraphicnamed{scalars//redcommutebraid2}
\InputIfFileExists{scalars//redcommutebraid2.tikz}{}{\input{./figures/scalars//redcommutebraid2.tikz}}
\endpgfgraphicnamed
\end{equation}
\end{lem}

With the 3D isotopy of diagrams in a braided autonomous category, we derive the inversely braided version of (B2).
\begin{lem}\label{lm:b2braid2nd}
The braided bialgebra rule holds with the inverse braiding:
\begin{equation}\label{b2braid2nd}
\beginpgfgraphicnamed{scalars//b2braid2ndscirc}
\InputIfFileExists{scalars//b2braid2ndscirc.tikz}{}{\input{./figures/scalars//b2braid2ndscirc.tikz}}
\endpgfgraphicnamed
 \end{equation}
 \end{lem}

 \begin{proof}
\begin{align*}
\beginpgfgraphicnamed{scalars//b2braid2ndprf_v2}
\InputIfFileExists{scalars//b2braid2ndprf_v2.tikz}{}{\input{./figures/scalars//b2braid2ndprf_v2.tikz}}
\endpgfgraphicnamed
\end{align*}
We flip the first diagram but keep the linear order of the edges entering and exiting, with respect to the 3D isotopy. Then we use equations  (\ref{eq:b2braid}), (\ref{eq:greencommuteflip}) and (\ref{greenredcommuteflipswapbraid}).
\end{proof}

\begin{lem}\label{braidissymm}
 The braiding is in fact symmetric.
\end{lem}

\begin{proof}
Begin by rewriting the diagram until the braided bialgebra rule can be applied, using the Hopf law and the spider rules:
\begin{align*}
\beginpgfgraphicnamed{scalars//braidsymmetry}
\InputIfFileExists{scalars//braidsymmetry.tikz}{}{\input{./figures/scalars//braidsymmetry.tikz}}
\endpgfgraphicnamed
\end{align*}
We then apply the inversely braided bialgebra rule and reverse the initial rewrite steps.
\end{proof}

\begin{thm}
 When working in a braided autonomous category, the rules in Figure \ref{figure1} with (B2$^\prime$) replaced by (B2$^{\prime\prime}$),  the rule \eqref{eq:s2prime}, and the loop rule (R1) are complete for the stabilizer \ZX-calculus.
\end{thm}

\begin{proof} The idea is to prove that the braiding is self inverse, meaning the category we are working in is actually symmetric monoidal:
\begin{equation}
\beginpgfgraphicnamed{scalars//braidselfinvb}
\InputIfFileExists{scalars//braidselfinvb.tikz}{}{\input{./figures/scalars//braidselfinvb.tikz}}
\endpgfgraphicnamed
\end{equation}
The proof of this equation is given in Lemma \ref{braidissymm}.

Once we know we have a symmetric monoidal category, we can show commutativity of green copy and co-copy as well as their colour-swapped versions by Lemmas \ref{greencommute2}, \ref{greencommuteflip}, \ref{redcommuteflipswap}, \ref{greeninvcommute}, and \ref{greenredcommuteflipswap}. Along with (S3$'$), we obtain the symmetry of cap and cup. Therefore, we come back to the situation described in the previous subsection: working in a compact closed category.
\end{proof}

\section{Conclusion and perspectives}

The stabilizer \ZX-calculus has a complete set of rewrite rules, which allow any equality that can be derived using matrices to also be derived graphically.
We introduce a simplified but still complete version of the stabilizer \ZX-calculus with significantly fewer rewrite rules. In particular, many rules obtained from others by swapping colours and/or flipping diagrams upside-down are no longer assumed. Our aim is to minimise the axioms of the language in order to pinpoint the fundamental structures of quantum mechanics, and also simplify the development and the efficiency of automated  tools for quantum reasoning, like Quantomatic \cite{quanto}.

Among the nine remaining rules of the language, only two are not proved to be necessary, although we know that at least one of them is. The problem of the minimality of the language is left as an open question and can essentially be phrased as follows: do the rules of the language (without the (S3$'$R) rule) force  the  two compact structures, induced by the red and green generators respectively, to coincide?

The simplified stabilizer \ZX-calculus can also serve as a backbone for further developments, in particular concerning the full calculus (allowing arbitrary angles). Several rules we showed to be derivable in the stabilizer \ZX-calculus are also derivable in the full \ZX-calculus: e.g.\ (ZS), which is valid for arbitrary angles, and (K1). The derivation of (K2) on the other hand is valid for the stabilizer fragment only.  Recently, new rules, including the so-called  supplementarity, have been proved to be necessary for the (full) \ZX-calculus \cite{PW15,CycloSupp} and in particular  for the $\pi/4$-fragment of the \ZX-calculus, which corresponds to the so called Clifford+T quantum mechanics. Even if supplementarity and (K2) rules can be derived in the stabilizer \ZX-calculus, a future project is to establish a simple, possibly minimal, set of axioms for the stabilizer \ZX-calculus which contains the rules known to be necessary for arbitrary angles (like supplementarity or (K2)), while avoiding rules which are in some sense specific to the $\pi/2$ fragment, e.g.\ (EU).

The fragment of \ZX-calculus made of the diagrams involving angles multiple of $\pi$ only, is known to be complete for the real stabilizer quantum mechanics \cite{duncan_pivoting_2014}, which is the basis of a full language for real quantum mechanics \cite{Ycalculus}. A perspective is to provide a simplified version of the real stabilizer \ZX-calculus, in particular considering the rules for which we fail to prove the necessity for the stabilizer \ZX-calculus.

We have also proved that the meta-rule `only the connectivity matters' can be derived from the rules of the language together with 3D-isotopy. The latter means that the ambient category is a braided autonomous category which additionally satisfies the Reidemeister rule (R1). We leave as an open question the necessity of the (R1) rule for deriving the connectivity meta-rule. The emergence of braided categories in this context opens new avenues for considering fermionic quantum mechanics \cite{panangaden2010categorical,davydov2013braided}.

A future step would be to extend the search for minimal complete rule sets to the Clifford+T fragment \cite{jeandel_complete_2017} or the full \ZX-calculus \cite{ng_universal_2017}.

\section*{Acknowledgements}

The authors would like to thank Bob Coecke, Ross Duncan, Emmanuel Jeandel, Aleks Kissinger, Kang Feng Ng and Renaud Vilmart for valuable discussions.
We also thank the anonymous reviewers for their comments.

QW acknowledges funding from R\'egion Lorraine,  EPSRC IAA in collaboration with Cambridge Quantum Computing Ltd., and AFOSR grant FA2386-18-1-4028. MB has received funding from EPSRC via grant EP/L021005/1 and from the European Research Council under the European Union's Seventh Framework Programme (FP7/2007-2013) ERC grant agreement no.\ 334828.
The paper reflects only the authors' views and not the views of the ERC or the European Commission.
The European Union is not liable for any use that may be made of the information contained therein. No new data were created during this study. SP acknowledges support from the projects ANR-17-CE25-0009
SoftQPro, ANR-17-CE24-0035 VanQuTe, PIA-GDN/Quantex, and LUE / UOQ.

\bibliographystyle{alpha}
\bibliography{refs}

\newcommand{\etalchar}[1]{$^{#1}$}
\begin{thebibliography}{JPVW17}

\bibitem[Bac14a]{backens_zx-calculus_2013}
Miriam Backens.
\newblock The {ZX}-calculus is complete for stabilizer quantum mechanics.
\newblock {\em New Journal of Physics}, 16(9):093021, September 2014.

\bibitem[Bac14b]{backens_zx-calculus_2014}
Miriam Backens.
\newblock The {ZX}-calculus is complete for the single-qubit {Clifford}+{T}
  group.
\newblock {\em Electronic Proceedings in Theoretical Computer Science},
  172:293--303, December 2014.

\bibitem[Bac15]{backens_making_2015}
Miriam Backens.
\newblock Making the stabilizer {ZX}-calculus complete for scalars.
\newblock {\em Electronic Proceedings in Theoretical Computer Science},
  195:17--32, November 2015.

\bibitem[BPW17]{BPW16}
Miriam Backens, Simon Perdrix, and Quanlong Wang.
\newblock A {Simplified} {Stabilizer} {ZX}-calculus.
\newblock {\em EPTCS}, 236:1--20, January 2017.

\bibitem[CD11]{coecke_interacting_2011}
Bob Coecke and Ross Duncan.
\newblock Interacting quantum observables: categorical algebra and
  diagrammatics.
\newblock {\em New Journal of Physics}, 13(4):043016, April 2011.

\bibitem[CJPV19]{carette2019completeness}
Titouan Carette, Emmanuel Jeandel, Simon Perdrix, and Renaud Vilmart.
\newblock Completeness of graphical languages for mixed states quantum
  mechanics.
\newblock In {\em International Colloquium on Automata, Languages, and
  Programming (ICALP'19)}, 2019.

\bibitem[CK17]{coecke_picturing_2017}
Bob Coecke and Aleks Kissinger.
\newblock {\em Picturing {Quantum} {Processes}: {A} {First} {Course} in
  {Quantum} {Theory} and {Diagrammatic} {Reasoning}}.
\newblock Cambridge University Press, 2017.

\bibitem[DP09]{DP09}
Ross Duncan and Simon Perdrix.
\newblock Graph states and the necessity of {Euler} decomposition.
\newblock In {\em Mathematical Theory and Computational Practice}, volume 5635,
  pages 167--177. Springer Berlin Heidelberg, 2009.

\bibitem[DP14]{duncan_pivoting_2014}
Ross Duncan and Simon Perdrix.
\newblock Pivoting makes the {ZX}-calculus complete for real stabilizers.
\newblock {\em Electronic Proceedings in Theoretical Computer Science},
  171:50--62, December 2014.

\bibitem[DR13]{davydov2013braided}
Alexei Davydov and Ingo Runkel.
\newblock A braided monoidal category for symplectic fermions.
\newblock {\em Symmetries and Groups in Contemporary Physics}, 11:399, 2013.

\bibitem[Got97]{gottesman_stabilizer_1997}
Daniel Gottesman.
\newblock {\em Stabilizer {Codes} and {Quantum} {Error} {Correction}}.
\newblock PhD thesis, Caltech, May 1997.

\bibitem[Got98]{gottesman_heisenberg_1998}
Daniel Gottesman.
\newblock The {Heisenberg} representation of quantum computers.
\newblock In {\em Proceedings of the XXII International Colloquium on Group
  Theoretical Methods in Physics}, July 1998.
\newblock \href{http://arxiv.org/abs/quant-ph/9807006}{arXiv:quant-ph/9807006}.

\bibitem[HNW18]{HNW}
Amar Hadzihasanovic, Kang~Feng Ng, and Quanlong Wang.
\newblock Two complete axiomatisations of pure-state qubit quantum computing.
\newblock In {\em Proceedings of the 33rd Annual ACM/IEEE Symposium on Logic in
  Computer Science}, LICS '18, pages 502--511, New York, NY, USA, 2018. ACM.

\bibitem[JLW18]{jaffe2018holographic}
Arthur Jaffe, Zhengwei Liu, and Alex Wozniakowski.
\newblock Holographic software for quantum networks.
\newblock {\em Science China Mathematics}, 61(4):593--626, 2018.

\bibitem[JPV18a]{jeandel_complete_2017}
Emmanuel Jeandel, Simon Perdrix, and Renaud Vilmart.
\newblock A complete axiomatisation of the {ZX}-calculus for {C}lifford+{T}
  quantum mechanics.
\newblock In {\em Proceedings of the 33rd Annual ACM/IEEE Symposium on Logic in
  Computer Science}, LICS '18, pages 559--568, New York, NY, USA, 2018. ACM.

\bibitem[JPV18b]{JPV-universal}
Emmanuel Jeandel, Simon Perdrix, and Renaud Vilmart.
\newblock Diagrammatic reasoning beyond {C}lifford+{T} quantum mechanics.
\newblock In {\em Proceedings of the 33rd Annual ACM/IEEE Symposium on Logic in
  Computer Science}, LICS '18, pages 569--578, New York, NY, USA, 2018. ACM.

\bibitem[JPV18c]{Ycalculus}
Emmanuel Jeandel, Simon Perdrix, and Renaud Vilmart.
\newblock Y-calculus: A language for real matrices derived from the
  zx-calculus.
\newblock In Bob Coecke and Aleks Kissinger, editors, {\em {\textrm Proceedings
  14th International Conference on} Quantum Physics and Logic, {\textrm
  Nijmegen, The Netherlands, 3-7 July 2017}}, volume 266 of {\em Electronic
  Proceedings in Theoretical Computer Science}, pages 23--57. Open Publishing
  Association, 2018.

\bibitem[JPV19]{ZXNormalForm}
Emmanuel Jeandel, Simon Perdrix, and Renaud Vilmart.
\newblock A generic normal form for zx-diagrams and application to the rational
  angle completeness.
\newblock In {\em Proceedings of the 34th Annual ACM/IEEE Symposium on Logic in
  Computer Science (LICS)}, 2019.

\bibitem[JPVW17]{CycloSupp}
Emmanuel Jeandel, Simon Perdrix, Renaud Vilmart, and Quanlong Wang.
\newblock {ZX}-calculus: Cyclotomic supplementarity and incompleteness for
  {C}lifford+{T} quantum mechanics.
\newblock {\em 42nd International Symposium on Mathematical Foundations of
  Computer Science (MFCS), arXiv preprint arXiv:1702.01945 [quant-ph]}, 2017.

\bibitem[KMF{\etalchar{+}}]{quanto}
Aleks Kissinger, Alex Merry, Ben Frot, Bob Coecke, David Quick, Lucas Dixon,
  Matvey Soloviev, Ross Duncan, and Vladimir Zamdzhiev.
\newblock Quantomatic.
\newblock \url{https://quantomatic.github.io/}.
\newblock Accessed September 2018.

\bibitem[NC10]{nielsen_quantum_2010}
Michael~A. Nielsen and Isaac~L. Chuang.
\newblock {\em {Quantum Computation and Quantum Information}}.
\newblock Cambridge University Press, Cambridge, 2010.

\bibitem[NW17]{ng_universal_2017}
Kang~Feng Ng and Quanlong Wang.
\newblock A universal completion of the {ZX}-calculus.
\newblock {\em arXiv:1706.09877 [quant-ph]}, June 2017.

\bibitem[PP10]{panangaden2010categorical}
Prakash Panangaden and {\'E}ric~Oliver Paquette.
\newblock A categorical presentation of quantum computation with anyons.
\newblock In {\em New structures for Physics}, pages 983--1025. Springer, 2010.

\bibitem[PW16]{PW15}
Simon Perdrix and Quanlong Wang.
\newblock {Supplementarity is Necessary for Quantum Diagram Reasoning}.
\newblock In {\em 41st International Symposium on Mathematical Foundations of
  Computer Science (MFCS 2016)}, volume~58 of {\em LIPIcs}, pages 76:1--76:14,
  2016.

\bibitem[RB01]{raussendorf_one-way_2001}
Robert Raussendorf and Hans~J. Briegel.
\newblock A one-way quantum computer.
\newblock {\em Physical Review Letters}, 86(22):5188--5191, May 2001.

\bibitem[Rei32]{reidemeister_knotentheorie_1932}
Kurt Reidemeister.
\newblock {\em Knotentheorie}.
\newblock Number~1 in Ergebnisse der {Mathematik} und ihrer {Grenzgebiete}.
  Julius Springer, Berlin, 1932.
\newblock English Translation: Knot Theory, B C S Associates (1983).

\bibitem[RV19]{reutter2019shaded}
David~J Reutter and Jamie Vicary.
\newblock Shaded tangles for the design and verification of quantum circuits.
\newblock {\em Proceedings of the Royal Society A}, 475(2224):20180338, 2019.

\bibitem[Sel07]{selinger_dagger_2007}
Peter Selinger.
\newblock Dagger {Compact} {Closed} {Categories} and {Completely} {Positive}
  {Maps}: ({Extended} {Abstract}).
\newblock {\em Electronic Notes in Theoretical Computer Science},
  170(0):139--163, March 2007.

\bibitem[Sel10]{selinger_survey_2010}
Peter Selinger.
\newblock A {Survey} of {Graphical} {Languages} for {Monoidal} {Categories}.
\newblock In Bob Coecke, editor, {\em New {Structures} for {Physics}}, number
  813 in Lecture {Notes} in {Physics}, pages 289--355. Springer Berlin
  Heidelberg, 2010.

\bibitem[Vil19]{euler-zx}
Renaud Vilmart.
\newblock A near-optimal axiomatisation of {ZX}-calculus for pure qubit quantum
  mechanics.
\newblock In {\em Proceedings of the 34th Annual ACM/IEEE Symposium on Logic in
  Computer Science (LICS)}, 2019.

\end{thebibliography}

\end{document}